\documentclass[a4paper,10pt]{article}
\usepackage{amsthm,amssymb, amsmath}
\usepackage{graphicx}
\usepackage{graphics}
\usepackage{times}
\usepackage{amsmath}
\usepackage{amsfonts}
\usepackage{color}

\newtheorem{Def}{Definition}
\newtheorem{Thm}{Theorem}
\newtheorem{Lem}{Lemma}
\newcommand{\mb}{\mathbf}
\newcommand{\mat}[2][ccccccccccccccccccccccccccc]{\left( \begin{array}
{#1}#2 \\ \end{array} \right)}

\begin{document}

\title{\textbf{Rate Region of the Gaussian Scalar-Help-Vector Source-Coding Problem}}


\author{Md Saifur Rahman and Aaron B. Wagner$^{\footnote{Both authors are with the School of Electrical and Computer Engineering, Cornell University, Ithaca, NY 14853 USA. (Email: mr534@cornell.edu, wagner@ece.cornell.edu.)}}$}

\maketitle

\begin{abstract}
We determine the rate region of the Gaussian scalar-help-vector source-coding problem under a covariance matrix distortion constraint.  The rate region is achieved by a Gaussian achievable scheme. We introduce a novel outer bounding technique to establish the converse of the main result. Our approach is based on lower bounding the problem with a potentially reduced dimensional problem by projecting the main source and imposing the distortion constraint in certain directions determined by the optimal Gaussian scheme. We also prove several properties that the optimal solution to the point-to-point rate-distortion problem for a vector Gaussian source under a covariance matrix distortion constraint satisfies. These properties play an important role in our converse proof. We further establish an outer bound to the rate region of the more general problem in which there are distortion constraints on both the sources. The outer bound is partially tight in general. We also study its tightness in some nontrivial cases.
\end{abstract}


\section{Introduction}

We consider the Gaussian scalar-help-vector source-coding problem (sometimes referred to as the one-helper problem). The setup of the problem is shown in Fig. 1. The first encoder observes a vector Gaussian source that is correlated with a scalar Gaussian source observed by the second encoder. The encoders separately send messages about their observations to the decoder at rates $R_1$ and $R_2$, respectively. The decoder uses both messages to estimate the first encoder's observations such that a certain distortion constraint on the average error covariance matrix of the estimate is satisfied. The goal is to determine the rate region of the problem, which is the set of all rate pairs $(R_1,R_2)$ that allow us to satisfy the distortion constraint.

\begin{figure}[htp]
\centering
  \includegraphics[width=3.0in]{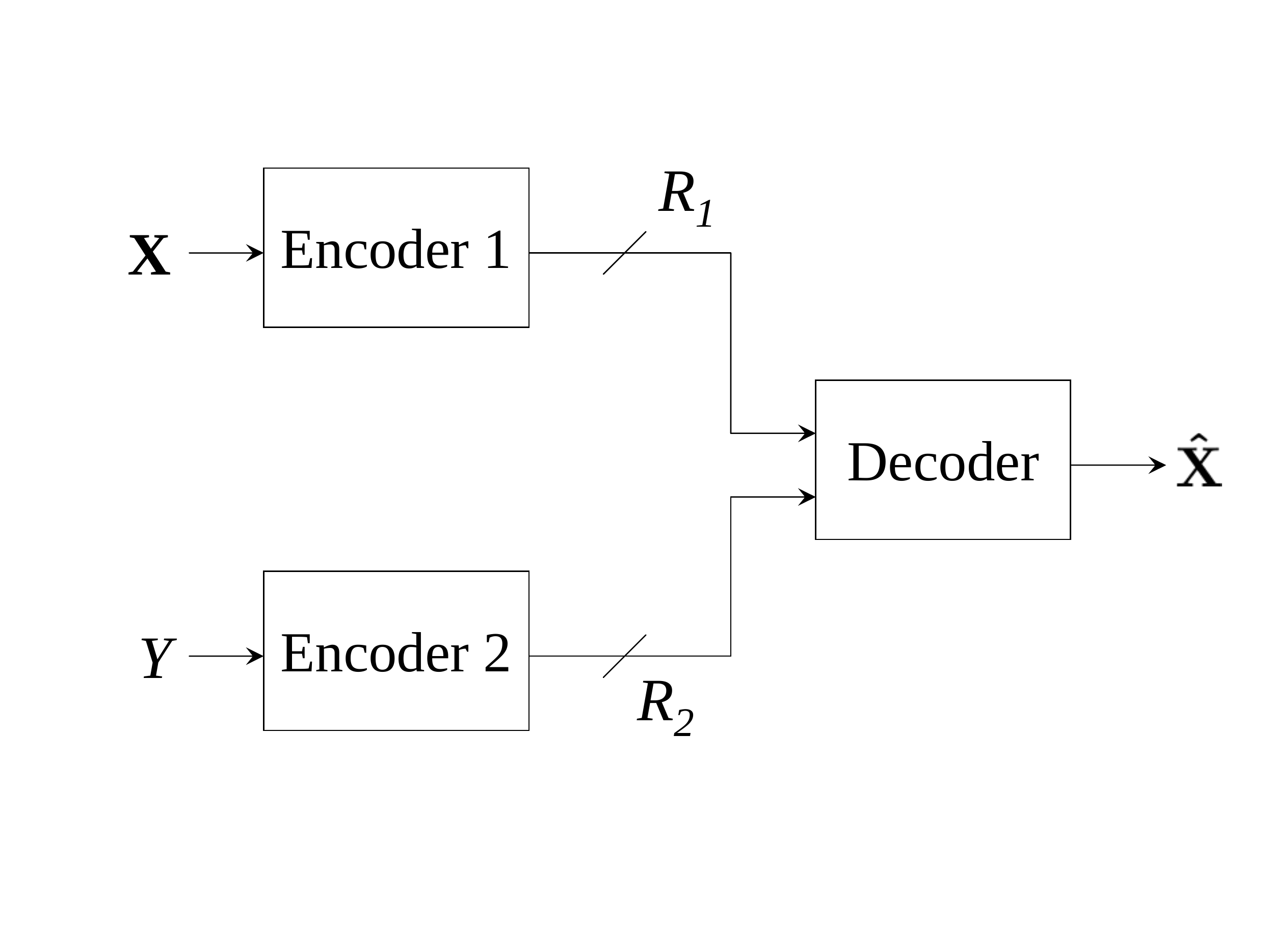}
\caption{The Gaussian scalar-help-vector source-coding problem}\label{fig:Fig1}
\end{figure}

Multiterminal source coding have received considerable attention. The study started with the work by Slepian and Wolf \cite{Slepian}. They considered a lossless problem in which two correlated memoryless sources must be reproduced by the decoder with arbitrarily small error probability. It immediately prompted researchers to extend Slepian and Wolf's result to lossy problems where the sources must be reconstructed with average distortions no more than a fixed amount. Wyner and Ziv \cite{WynerZiv} were the first who studied the lossy problem where the decoder reconstructs a source within allowable distortion with the help of side information about the source. Since then, the problem has been extended in several directions by Berger \cite{Berger}, Tung \cite{Tung}, Berger \emph{et al}. \cite{Zhang}, Oohama \cite{Oohama}, Viswanathan and Berger \cite{Viswanathan}, Wagner \emph{et al.} \cite{Wagner}, Liu and Viswanath \cite{Lui} and several others.

Oohama \cite{Oohama} studied the one-helper problem when both the sources are scalar Gaussian and gave a complete characterization of the rate region. He used the conditional entropy power inequality (EPI), which was also used by Bergmans to determine the capacity region of the scalar Gaussian degraded broadcast channel \cite{Bergmans}, to establish the converse for the one-helper problem. Oohama's proof suggests a connection between these two problems. Recently, Weingarten \emph{et al.} \cite{Wein} extended Bergmans's result to the vector case and determined the capacity region of the vector Gaussian multiple-input-multiple-output broadcast channel by introducing the idea of enhancement. It is natural to expect the enhancement idea to be useful for the vector extension of Oohama's result. In this context, Liu and Viswanath \cite{Lui} studied the one-helper problem when both the sources are vector Gaussian and the distortion constraint is on the average error covariance matrix. They combined Oohama's converse arguments with Weingarten's enhancement idea to obtain a lower bound. However, their lower bound is not tight in general because the steps in Oohama's converse proof are in general not tight if the sources are vector Gaussian. In particular, the distortion constraint is not met with equality in general. In our earlier work \cite{Rahman}, we used the enhancement idea with an improved outer bounding technique to characterize a portion of boundary of the rate region for the vector Gaussian one-helper problem. However, the complete characterization of the rate region remains unknown.

\begin{figure}[htp]
\centering
  \includegraphics[width=3.0in]{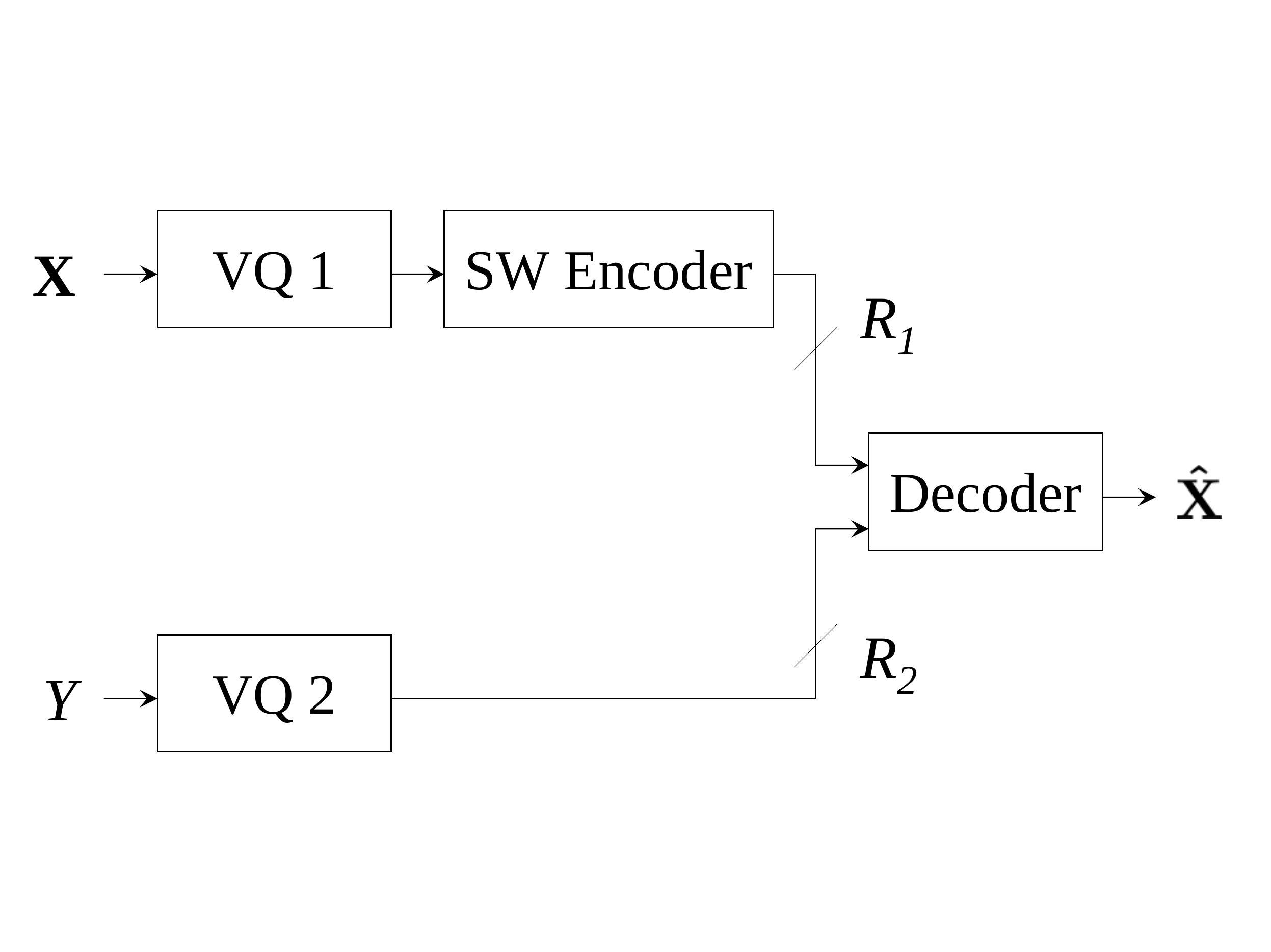}
\caption{A Gaussian achievable scheme}\label{fig:Fig2}
\end{figure}

We consider the simplest version of the problem that cannot be solved using existing techniques, namely that in which the primary source is a vector and the helper's observation is a scalar. For this Gaussian scalar-help-vector source-coding problem, we completely determine the rate region using a novel outer bounding technique. We find the optimal Gaussian solution and determine the set of directions in which this scheme meets the distortion constraint with equality. This set of directions is used to define a potentially reduced dimensional problem that lower bounds the original problem. We then proceed as Oohama did to obtain a lower bound to the reduced dimensional problem. The lower bound thus obtained is achieved by the Gaussian achievable scheme depicted in Fig. 2. In this scheme, the first encoder vector quantizes (VQ) its observation using a Gaussian test channel as in point-to-point rate-distortion theory. It then compresses the quantized values using Slepian-Wolf encoding. The second encoder just vector quantizes its observation using another Gaussian test channel. The decoder decodes the quantized values and estimates the observations of the first encoder using a minimum mean-squared error (MMSE) estimator.

We then study the more general problem in which there are distortion constraints on both the sources. Oohama \cite{Oohama} and Wagner \emph{et al.} \cite{Wagner} studied the scalar version of the problem and their work together characterizes the rate region completely. Following Oohama's technique \cite{Oohama}, one can obtain an outer bound to the rate region of the general source-coding by considering two one-helper relaxations, each of which can be solved. A portion of the resulting outer bound is tight in general, but the rest of it is fairly loose. We obtain an improved outer bound and give sufficient conditions for its tightness. Our approach is based on splitting the first encoder's rate into two parts. The first part of the rate is used to communicate information about the sufficient statistic for $Y$ given $\mb{X}$, and the second part is used to communicate information about the remaining components of $\mb{X}$. The first component is analyzed using results for the scalar problem with two distortion constraints
\cite{Oohama,Wagner}. The second part is analyzed using point-to-point rate-distortion theory.

Another contribution of this paper is that we establish several properties that the optimal solution to the point-to-point vector Gaussian source coding problem satisfies. The core optimization problem here is to maximize the $\log |\cdot|$ function over a set of positive semidefinite matrices that are no more than two positive definite matrices in a positive semidefinite sense. Since this is a convex optimization problem, its optimal solution must satisfy the necessary and sufficient Karush-Kuhn-Tucker (KKT) conditions \cite{Boyd}. By analyzing the KKT conditions, we arrive at several interesting properties that the optimal solution satisfies. These properties are used to prove the converse of our main result,
and they could prove useful elsewhere, even outside network information theory.

The rest of the paper is organized as follows. Section 2 explains the notations used in the paper. In Section 3, we present the mathematical formulation of the problem and summarize our main results. Section 4 is devoted to study the core optimization problem in the point-to-point rate-distortion theory for vector Gaussian sources under covariance matrix distortion constraint. In Section 5, we give the converse proof on our main result. Finally in Section 6, we study the generalization of problem in which there are distortion constraints on both the sources.

\section{Notations}
We use uppercase to denote random variables and vectors. Boldface is used to distinguish vectors from scalars. Arbitrary realizations of random variables and vectors are denoted in lowercase. Let $\left \{{\mb{X}_i} \right \}_{i=1}^{n}$ be an independent and identically distributed (i.i.d.) random process of random vectors. We use $\mb{X}^n$ to denote $\left \{{\mb{X}_i} \right \}_{i=1}^{n}$, and $\mb{x}^n$ to denote an arbitrary realization $\left \{{\mb{x}_i} \right \}_{i=1}^{n}$. The superscript $T$ denotes matrix transpose. We use $\sigma^2_Y$ and $\sigma^2_{Y|V}$ to denote the variance of $Y$ and the conditional variance of $Y$ given $V$, respectively. The covariance matrix of $\mb{X}$ is denoted by $\mb{K_X}$. The conditional covariance matrix of $\mb{X}$ given $\mb{Y}$ is denoted by $\mb{K}_{\mb{X}|\mb{Y}}$, and is defined as
\[
\mb{K}_{\mb{X}|\mb{Y}} = E \left [ \left (\mb{X} - E(\mb{X}|\mb{Y}) \right ) \left (\mb{X} - E(\mb{X}|\mb{Y}) \right )^T \right ].
\]
All vectors are column vectors, and are $m$-dimensional, unless otherwise stated. We use $\mb{I}_m$ to denote an $m \times m$ identity matrix. With a little abuse of notation, $\mb{0}$ is used to denote both zero vectors and zero matrices of any dimensions. For two positive semidefinite matrices $\mb{A}$ and $\mb{B}$, $\mb{A} \succcurlyeq \mb{B}$ ($\mb{A} \succ \mb{B}$) means that $\mb{A-B}$ is positive semidefinite (definite). Similarly, $\mb{A} \preccurlyeq \mb{B}$ ($\mb{A} \prec \mb{B}$) means that $\mb{B-A}$ is positive semidefinite (definite). All logarithms in this paper are to the base 2.

\section{Problem Formulation and Main Results}

\subsection{The Gaussian Scalar-Help-Vector Source-Coding Problem}
Let $\left \{(\mb{X}_i,Y_i) \right  \}_{i=1}^{n}$ be a sequence of i.i.d. zero-mean Gaussian random vectors. At each time $i$, $\mb{X}_i$ and $Y_i$ are jointly Gaussian with the covariance matrix $\mb{K_X}$ and the variance $\sigma^2_Y$, respectively. Without loss of generality, we can write
\begin{align*}
\mb{X}_i= \mb{a}Y_i+\mb{N}_i,
\end{align*}
where $\mb{a}$ is a vector and $\mb{N}_i$ is a zero-mean Gaussian random vector with the covariance matrix $\mb{K_{N}}$. Assume further that $Y_i$ is independent of $\mb{N}_i$. Therefore,
\[
\mb{K_X} = \mb{a} \mb{a}^T \sigma^2_Y + \mb{K_N}.
\]

The first encoder observes $\mb{X}^n$ and sends a message to the decoder using an encoding function
\begin{align*}
f_1^{(n)} : \mathbb{R}^{mn} \mapsto \left \{1,\dots,M_1^{(n)} \right \}.
\end{align*}
Analogously, the second encoder observes $Y^n$ and sends a message to the decoder using another encoding function
\begin{align*}
f_2^{(n)}  : \mathbb{R}^{n} \mapsto \left \{1,\dots,M_2^{(n)} \right \}.
\end{align*}
The decoder uses both received messages to estimate $\mb{X}^n$ using a decoding function
\begin{align*}
g^{(n)}  : \left \{1,\dots,M_1^{(n)} \right \} \times \left \{1,\dots,M_2^{(n)} \right \} \mapsto \mathbb{R}^{mn}.
\end{align*}
\begin{Def}
A rate-distortion vector $\left (R_1,R_2,\mb{D} \right )$, where $\mb{D}$ is a positive definite matrix, is \emph{achievable} for the Gaussian scalar-help-vector source-coding problem if there exists a block length $n$, encoders $f_1^{(n)} $ and $f_2^{(n)} $, and a decoder $g^{(n)} $ such that
\begin{align*}
R_i &\ge \frac{1}{n} \log M_i^{(n)}  \hspace {0.15 cm} \textrm{for all} \hspace {0.15 cm} i \in \{1,2\}, \hspace {0.15 cm} \textrm{and}\\
\mb{D} &\succcurlyeq \frac{1}{n} \sum_{i=1}^n E \left [ \left (\mb{X}_i - \hat {\mb{X}}_i \right ) \left (\mb{X}_i - \hat {\mb{X}}_i \right )^T \right ],
\end{align*}
where
\begin{align*}
\hat {\mb{X}}^n = g^{(n)} \left (f_1^{(n)} \left (\mb{X}^n \right ), f_2^{(n)}\left (Y^n \right ) \right ).
\end{align*}
Let $\hat {\mathcal{R}}$ be the closure of the set of all achievable rate-distortion vectors. Define
\[
\mathcal{R}\left (\mb{D} \right) = \left \{(R_1,R_2): (R_1,R_2,\mb{D}) \in \hat {\mathcal{R}} \right \}.
\]
We call $\mathcal{R} (\mb{D})$ the \emph{rate region} for the Gaussian scalar-help-vector source-coding problem.
\end{Def}
Since we are interested in the covariance matrix distortion constraint, without loss of generality we can restrict the decoding function to be the MMSE estimate of $\mb{X}^n$ based on the received messages. Therefore, $\hat {\mb{X}}^n$ can be written as
\begin{align*}
\hat {\mb{X}}^n = E \left [\mb{X}^n | f_1^{(n)} \left (\mb{X}^n \right ), f_2^{(n)}\left (Y^n \right )\right ].
\end{align*}
If $\mb{a} = 0$, then the problem reduces to the point-to-point vector Gaussian rate-distortion problem which can be solved using existing techniques. Therefore, we will assume that $\mb{a} \neq 0$ in the rest of the paper. We will assume further that $\mb{K_X}$ is strictly positive definite, since the case when $\mb{K_X}$ is singular can be handled by defining an equivalent problem in which the source covariance is strictly positive definite. We can do so by applying an invertible transformation. The details of the transformation is presented in Appendix A.
\subsection{Rate Region}
Let us define the following set
\begin{align*}
\mathcal{R}^{*}(\mb{D}) = \{(R_1,R_2) : R_1 \ge  \min_{\mb{K}} \hspace {0.1in} &\frac{1}{2} \log \frac{ \left |\mb{a} \mb{a}^{T} \sigma_{Y}^2 2^{-2 R_2} + \mb{K_{N}} \right |}{ \left|\mb{K} \right|} \\
\textrm{s. t.} \hspace{0.1in} &\mb{0} \preccurlyeq \mb{K} \preccurlyeq \mb{D}\\
&\mb{K} \preccurlyeq \mb{a} \mb{a}^{T} \sigma_{Y}^2 2^{-2 R_2} + \mb{K_N} \}.
\end{align*}
We then have the following theorem.
\begin{Thm}
For every positive definite matrix $\mb{D}$
\begin{align*}
\mathcal{R}(\mb{D}) = \mathcal{R}^{*}(\mb{D}).
\end{align*}
\end{Thm}
\subsection{A Gaussian Achievable Scheme} There is a natural \emph{Gaussian achievable scheme} depicted in Fig. 2 that is optimal for the problem described in the previous section. Similar schemes have been shown to be optimal for other important problems in Gaussian multiterminal source coding \cite{Oohama, Wagner, Wagner1, Wang, Berger, Zhang, Viswanathan}. We present an overview of the scheme here. The details of the scheme can be found in \cite{Zhang, Oohama}.

Let $\mathcal{S}(U,V)$ be the set of zero-mean jointly Gaussian random variables $U$ and $V$ such that
\begin{enumerate}
\item[(a)] $U, \mb{X}$, $Y$, and $V$ form a Markov chain $U \leftrightarrow \mb{X} \leftrightarrow Y \leftrightarrow V$, and
\item[(b)] $\mb{K}_{\mb{X}|U,V} \preccurlyeq \mb{D}$.
\end{enumerate}
Consider any $(U,V) \in \mathcal{S(U,V)}$ and a large block length $n$. Let $R_1^{'} = I(\mb{X};U)+\epsilon$, where $\epsilon > 0$. To construct the codebook for the first encoder, first generate $2^{n R_1^{'}}$ independent codewords $U^n$ randomly according to the marginal distribution of $U$, and then uniformly distribute them into $2^{n R_1}$ bins. The second encoder's codebook is constructed by generating $2^{n R_2}$ independent codewords $V^n$ randomly according to the marginal distribution of $V$.

Given a source sequence $\mb{X}^n$, the first encoder looks for a codeword $U^n$ that is jointly typical with $\mb{X}^n$, and sends the index $b$ of the bin in which $U^n$ belongs. The second receiver upon receiving $Y^n$, sends the index of the codeword $V^n$ that is jointly typical with $Y^n$. The decoder receives the two indices, then looks into the bin $b$ for a codeword $U^n$ that is jointly typical with $V^n$. The decoder can recover $U^n$ and $V^n$ with high probability as long as
\begin{align*}
R_1 &\ge I(\mb{X};U|V)\\
R_2 &\ge I(Y;V).
\end{align*}
The decoder then computes the MMSE estimate of the source $\mb{X}^n$ given the messages $U^n$ and $V^n$, and $(b)$ above guarantees that this estimate will satisfy the covariance matrix distortion constraint. Let
\begin{align*}
\mathcal{R}_G(\mb{D}) = \{(R_1,R_2) &: \hspace{0.05in} \textrm{there exists}\hspace{0.05in} (U,V) \in \mathcal{S}(U,V) \hspace{0.05in}\textrm{such that} \\
R_1 &\ge I(\mb{X};U|V) \\
R_2 &\ge I(Y;V)\}.
\end{align*}
The following lemma gives the achievable rate region by using this scheme.
\begin{Lem}
The \emph{Gaussian achievable scheme} achieves $\mathcal{R}_G(\mb{D})$, which satisfies
\[
\mathcal{R}_G(\mb{D}) = \mathcal{R}^{*}(\mb{D}).
\]
\end{Lem}
It immediately follows from the discussion above that the \emph{Gaussian achievable scheme} achieves $\mathcal{R}_G(\mb{D})$. The equality in Lemma 1 is proved later in Section 5 (problem $P_G$ and Lemma 3). Theorem 1 and Lemma 1 together imply that $\mathcal{R}_G(\mb{D})$ equals the rate region $\mathcal{R}\left (\mb{D} \right)$. In particular, this proves that the \emph{Gaussian achievable scheme} depicted in Fig. 2 is optimal for this problem.

The converse proof of Theorem 1 is deferred to Section 5. In the next Section, we study an optimization problem which appears in the converse proof of our main result.

\section{The Core Optimization Problem}
In this section, we study the core optimization problem in the point-to-point rate-distortion theory for a vector Gaussian source under a covariance matrix distortion constraint. In this setup, an i.i.d. zero-mean vector Gaussian source $\mb{Z}^n$ with a covariance matrix $\mb{K_Z}$ is observed by the encoder which sends a message to the decoder over a rate-constrained channel using a function
\begin{align*}
f^{(n)}  : \mathbb{R}^{mn} \mapsto \left \{1,\dots,M^{(n)} \right \}.
\end{align*}
The decoder uses the received message to give an estimate $\hat{\mb{Z}}^n$ of the source $\mb{Z}^n$ such that
\[
\frac{1}{n} \sum_{i=1}^n E\left [ \left (\mb{Z}_i - \hat {\mb{Z}}_i \right  ) \left (\mb{Z}_i - \hat {\mb{Z}}_i \right )^T \right ]  \preccurlyeq \mb{D_Z},
\]
where $\mb{D_Z}$ is a positive definite matrix. As explained in Section 3.1, we can assume without loss of generality that $\mb{K}_{\mb{Z}}$ is strictly positive definite and
\begin{align*}
\hat{\mb{Z}}^n &= E \left [\mb{Z}^n| f^{(n)} \left ( \mb{Z}^n \right) \right].
\end{align*}

In the single-letter form, the rate-distortion function of the source $\mb{Z}^n$ is given by the optimal value of the following optimization problem
\begin{align*}
\min_{U} \hspace {0.1in} &I(\mb{Z};U) \nonumber\\
\textrm{subject to} \hspace{0.1in} &\mb{K}_{\mb{Z}|U} \preccurlyeq \mb{D_Z} \nonumber.
\end{align*}
A Gaussian $U$ is the optimal solution to this problem because of the fact that the Gaussian distribution maximizes the differential entropy for a given covariance matrix. So, we just need to optimize the above optimization problem over all Gaussian distributions. Equivalently, we have the following matrix optimization problem
\begin{align}
F \left (\mb{D_Z}, \mb{K_Z} \right ) = \hspace{0.3in}\max \hspace {0.1in} &\log \left |\mb{K}_{\mb{Z}|U} \right | \nonumber\\
\hspace{0.2in} \textrm{subject to} \hspace{0.1in} &\mb{0} \preccurlyeq \mb{K}_{\mb{Z}|U} \preccurlyeq \mb{D_Z}\\
&\mb{K}_{\mb{Z}|U}\preccurlyeq \mb{K_Z}. \nonumber
\end{align}
Note that we need to impose an additional constraint
\[
\mb{K}_{\mb{Z}|U} \preccurlyeq \mb{K_Z}
\]
because of the fact that the conditional covariance is no more than the unconditional covariance in a positive semidefinite sense. Observe that if any one of the two constraints in (1) is inactive, then the other constraint will be met with equality and we will have a close form solution to the problem. If both constraints are active, then it is unlikely to obtain a close form solution in general. However, we can establish certain properties that the optimal solution satisfies. The rest of the section is devoted to establish these properties.

Since the objective of the optimization problem (1) is continuous and
\[
\left \{\mb{0} \preccurlyeq \mb{K}_{\mb{Z}|U} \preccurlyeq \mb{D_Z} \hspace{0.1in} \text{and} \hspace{0.1in} \mb{K}_{\mb{Z}|U} \preccurlyeq \mb{K_Z} \right \}
\]
is a compact set, there exists an optimal solution $\mb{K}^{*}_{\mb{Z}|U}$ to (1). Furthermore, since (1) is a convex optimization problem, we can get the following Lagrangian formulation
\begin{align*}
\max \hspace{0.2in} \log \left |\mb{K}_{\mb{Z}|U} \right | + \textrm{Trace} \left \{\mb{K}_{\mb{Z}|U}\mb{\Lambda}  - \mb{K}_{\mb{Z}|U}\mb{M}_1 - \mb{K}_{\mb{Z}|U} \mb{M}_2 \right \},
\end{align*}
where $\mb {\Lambda}, \mb{M}_1$ and  $\mb{M}_2$ are positive semidefinite Lagrange multiplier matrices corresponding to the constraints $\mb{K}_{\mb{Z}|U} \succcurlyeq \mb{0}$, $\mb{K}_{\mb{Z}|U} \preccurlyeq \mb{D_Z}$ and $\mb{K}_{\mb{Z}|U} \preccurlyeq \mb{K_Z}$, respectively. Then $\mb{K}_{\mb{Z}|U}^{*}$ must satisfy the following necessary and sufficient KKT conditions \cite{Boyd}
\begin{align}
{\mb{K}_{\mb{Z}|U}^{*-1}} + \mb{\Lambda}^{*} - \mb{M}_1^{*} - \mb{M}_2^{*} &= \mb{0}, \\
\mb{K}_{\mb{Z}|U}^{*}\mb{\Lambda}^{*} &= \mb{0}, \\
\left (\mb{D_Z} - \mb{K}_{\mb{Z}|U}^{*} \right )\mb{M}_1^{*} &= \mb{0}, \\
\left (\mb{K_Z} - \mb{K}_{\mb{Z}|U}^{*}\right  )\mb{M}_2^{*} &= \mb{0}, \\
\mb{\Lambda}^{*}, \mb{M}_1^{*}, \mb{M}_2^{*} &\succcurlyeq \mb{0}.
\end{align}
Observe that in the optimization problem $(1)$, the constraint $\mb{K}_{\mb{Z}|U} \succcurlyeq \mb{0}$ is never active, so
\begin{align}
\mb{\Lambda}^{*}=\mb{0}.
\end{align}
Since $\log |\cdot|$ is strictly convex over the domain of positive definite matrices, $\mb{K}_{\mb{Z}|U}^{*}$ is the unique maximizer of the problem. Let
\begin{align*}
\mathcal{Q} = \left \{ \left (\mb{M}_1^{*},\mb{M}_2^{*} \right )\right \}
\end{align*}
be the set of all pairs $\left (\mb{M}_1^{*},\mb{M}_2^{*} \right )$ of Lagrange multiplier matrices that satisfy the KKT conditions. Consider any sequence of pairs $\left (\mb{M}_1^{*},\mb{M}_2^{*} \right )_n$ in $\mathcal{Q}$. By the continuity of the KKT conditions, it follows that the limit of this sequence belongs to $\mathcal{Q}$. Therefore, $\mathcal{Q}$ is a closed set. Equations (2) and (7) together imply that every pair $\left (\mb{M}_1^{*},\mb{M}_2^{*} \right )$ in $\mathcal{Q}$ is such that
\[
\mb{M}_1^{*} + \mb{M}_2^{*} = {\mb{K}_{\mb{Z}|U}^{*-1}},
\]
which is a fixed matrix. Hence, the set $\mathcal{Q}$ is bounded. We thus conclude that  $\mathcal{Q}$ is a compact set. This along with the continuity of the Trace$(\cdot)$ function imply that there exists $\left (\bar {\mb{M}}_1,\bar {\mb{M}}_2 \right )$ in $\mathcal{Q}$ that solves the optimization problem
\begin{align}
\min \hspace {0.1in} &\textrm{Trace}\left (\mb{M}_1^{*} \right )\nonumber\\
\textrm{subject to} \hspace{0.1in} &\left (\mb{M}_1^{*},\mb{M}_2^{*} \right ) \in \mathcal{Q} .
\end{align}
Since $\bar {\mb{M}}_1$ and $\bar {\mb{M}}_2$ are positive semidefinite, we can write their spectral decompositions as
\begin{align}
\bar {\mb{M}}_1 &= \sum_{i=1}^{r} \lambda_i \mb{s}_i \mb{s}_i^{T}  \\
\bar {\mb{M}}_2 &= \sum_{i=1}^{l} \gamma_i \mb{t}_i \mb{t}_i^{T},
\end{align}
where
\begin{enumerate}
\item[(a)] $0 \le r,l \le m$,
\item[(b)] $\lambda_i, \gamma_j > 0,$ for all $i \in \{1,\dots,r\}$ and for all $j \in \{1,\dots,l\}$, and
\item[(c)] $\{\mb{s}_i\}_{i=1}^{r}$ and $\{\mb{t}_i\}_{i=1}^{l}$ are sets of orthonormal vectors.
\end{enumerate}
We have from (4), (5), (9) and (10) that
\begin{align*}
\left (\mb{D_Z} - {\mb{K}_{\mb{Z}|U}^{*}} \right ) \sum_{i=1}^{r} \lambda_i \mb{s}_i \mb{s}_i^{T} &=\mb{0}, \\
\left (\mb{K_Z} - {\mb{K}_{\mb{Z}|U}^{*}} \right ) \sum_{i=1}^{l} \gamma_i \mb{t}_i \mb{t}_i^{T} &=\mb{0},
\end{align*}
which imply that
\begin{align}
\left (\mb{D_Z}  - {\mb{K}_{\mb{Z}|U}^{*}} \right ) \mb{s}_i &= \mb{0}, \hspace{0.2in}\forall i \in \left \{1,\dots,r \right \}, \\
\left (\mb{K_Z} - {\mb{K}_{\mb{Z}|U}^{*}} \right) \mb{t}_i &= \mb{0}, \hspace{0.2in}\forall i \in \left \{1,\dots,l \right \}.
\end{align}
Define the matrices
\begin{align*}
\mb{S} &= \left [\sqrt{\lambda_1}\mb{s}_1,\sqrt{\lambda_2}\mb{s}_2,\dots,\sqrt{\lambda_r}\mb{s}_r \right ] \\
\mb{T} &= \left [\sqrt{\gamma_1}\mb{t}_1,\sqrt{\gamma_2}\mb{t}_2,\dots,\sqrt{\gamma_l}\mb{t}_l \right ].
\end{align*}
Let $\mb{B}$ be an $m \times m$ positive definite matrix.
\begin{Def}
A non-zero $m \times p$ matrix $\mb{U}$ is $\mb{B}$\emph{-orthogonal} if
\begin{align*}
\mb{U}^T \mb{B} \mb{U} = \mb{I}_p.
\end{align*}
\end{Def}
\begin{Def}
A non-zero $m \times p$ matrix $\mb{U}$ and a non-zero $m \times q$ matrix $\mb{V}$ are \emph{cross} $\mb{B}$\emph{-orthogonal} if
\begin{align*}
\mb{U}^T \mb{B} \mb{V} = \mb{0}.
\end{align*}
\end{Def}
We have the following theorem about the optimal solution to the optimization problem (1).
\begin{Thm}
\begin{enumerate}
\item[(a)] If $r>0$, then $\mb{S}^T \left (\mb{K_Z} - \mb{K}_{\mb{Z}|U}^{*} \right) \mb{S}$ is strictly positive definite,
\item[(b)] $[\mb{S,T}]$ is an invertible matrix,
\item[(c)] $[\mb{S,T}]$ is $\mb{K}_{\mb{Z}|U}^{*}$\emph{-orthogonal},
\item[(d)] $\mb{S}$ is ${\mb{D_Z}}$\emph{-orthogonal},
\item[(e)] $\mb{T}$ is ${\mb{K_Z}}$\emph{-orthogonal},
\item[(f)] $\mb{S}$ and $\mb{T}$ are \emph{cross} ${\mb{D_Z}}$\emph{-orthogonal},
\item[(g)] $\mb{S}$ and $\mb{T}$ are \emph{cross} ${\mb{K_Z}}$\emph{-orthogonal}.
\end{enumerate}
\end{Thm}
\begin{proof}
For part (a), it suffices to show that $\mb{S}^T \left (\mb{K_Z} - \mb{K}_{\mb{Z}|U}^{*} \right) \mb{S}$ is non-singular.  Suppose otherwise that it is singular. Then there exists $\mb{0} \neq \mb{e} \in \mathbb{R}^r$ such that
\[
\mb{e}^T\mb{S}^T \left (\mb{K_Z} - \mb{K}_{\mb{Z}|U}^{*} \right) \mb{S}\mb{e} = 0.
\]
Let $\mb{w} = \mb{S e}$. We then have
\begin{align}
\mb{w} = \sum_{i=1}^{r} e_i \mb{s}_i,
\end{align}
where $e_i$ is the $i$-th component of $\mb{e}$, and
\begin{align}
\left (\mb{K_Z} - \mb{K}_{\mb{Z}|U}^{*} \right) \mb{w} = \mb{0}.
\end{align}
Let
\begin{align}
\lambda_{min} = \min \left \{ \lambda_1, \lambda_2, \dots, \lambda_r \right \}
\end{align}
and pick any $\epsilon$ such that
\begin{align}
0 < \epsilon \le \frac{\lambda_{min}}{\|\mb{w}\|^2}.
\end{align}
Consider any $\mb{0} \neq \mb{z} \in \mathbb{R}^m$. Let $\mb{z_S}$ be the projection of $\mb{z}$ on span$\{\mb{S}\}$ and $\mb{z}_{\mb{S}^{\perp}}$ be the projection of $\mb{z}$ on the space orthogonal to span$\{\mb{S}\}$. We then have
\begin{align*}
\mb{z} = \mb{z_S}+\mb{z}_{\mb{S}^{\perp}}.
\end{align*}
Now,
\begin{align}
\mb{z}^{T}\bar {\mb{M}}_1 \mb{z} &= \left (\mb{z}_{\mb{S}}+\mb{z}_{\mb{S}^{\perp}} \right)^T \bar {\mb{M}}_1 \left (\mb{z}_{\mb{S}}+\mb{z}_{\mb{S}^{\perp}} \right) \nonumber\\
&= \mb{z}_{\mb{S}}^{T} {\bar {\mb{M}}}_1 \mb{z}_{\mb{S}}+\mb{z}_{\mb{S}}^{T} {\bar {\mb{M}}}_1 \mb{z}_{\mb{S}^{\perp}}+\mb{z}_{\mb{S}^{\perp}}^{T} {\bar {\mb{M}}}_1 \mb{z}_{\mb{S}}+\mb{z}_{\mb{S}^{\perp}}^{T} {\bar {\mb{M}}}_1 \mb{z}_{\mb{S}^{\perp}} \nonumber\\
&= \mb{z}_{\mb{S}}^{T} {\bar {\mb{M}}}_1 \mb{z}_{\mb{S}},
\end{align}
where (17) follows because
\[
{\bar {\mb{M}}}_1 \mb{z}_{\mb{S}^{\perp}}=\mb{0}.
\] Similarly,
\begin{align}
\mb{w}^{T} \mb{z} = \mb{w}^{T} \mb{z}_{\mb{S}}.
\end{align}
We now have
\begin{align}
\mb{z}^{T} \left (\bar {\mb{M}}_1 - \epsilon \mb{w} \mb{w}^{T}\right )\mb{z} &= \mb{z}^{T}\bar {\mb{M}}_1 \mb{z} - \epsilon \left (\mb{w}^{T} \mb{z} \right )^2 \nonumber\\
&= \mb{z}_{\mb{S}}^{T} {\bar {\mb{M}}}_1 \mb{z}_{\mb{S}} - \epsilon \left (\mb{w}^{T} \mb{z}_{\mb{S}} \right )^2 \\
&\ge \mb{z}_{\mb{S}}^{T} {\bar {\mb{M}}}_1 \mb{z}_{\mb{S}} - \epsilon \|\mb{w}\|^2 \|\mb{z}_{\mb{S}}\|^2 \\
&= \sum_{i=1}^{r} \lambda_i \left (\mb{s}_i^{T} \mb{z}_{\mb{S}} \right )^2  - \epsilon \|\mb{w}\|^2 \|\mb{z}_{\mb{S}}\|^2 \\
&\ge \lambda_{min} \sum_{i=1}^{r} \left (\mb{s}_i^{T} \mb{z}_{\mb{S}} \right )^2  - \epsilon \|\mb{w}\|^2 \|\mb{z}_{\mb{S}}\|^2 \\
&= \lambda_{min} \|\mb{z}_{\mb{S}}\|^2  - \epsilon \|\mb{w}\|^2 \|\mb{z}_{\mb{S}}\|^2 \\
&= \|\mb{z}_{\mb{S}}\|^2  \left (\lambda_{min}  - \epsilon \|\mb{w}\|^2 \right ) \nonumber \\
&\ge 0,
\end{align}
where
\begin{enumerate}
\item[(19)] follows from (17) and (18),
\item[(20)] follows from the Cauchy-Schwartz Inequality,
\item[(21)] follows from (9),
\item[(22)] follows from (15),
\item[(23)] follows because
\[
\|\mb{z}_{\mb{S}}\|^2 = \sum_{i=1}^{r} \left (\mb{s}_i^{T} \mb{z}_{\mb{S}} \right )^2, \hspace{0.05in} \textrm{and}
\]
\item[(24)] follows from (16).
\end{enumerate}
This proves that $\bar {\mb{M}}_1 - \epsilon \mb{w} \mb{w}^{T}$ is a positive semidefinite matrix. Let us define the matrices
\begin{align*}
\tilde {\mb{M}}_1 &= \bar {\mb{M}}_1 - \epsilon \mb{w} \mb{w}^{T} \\
\tilde {\mb{M}}_2 &= \bar {\mb{M}}_2 + \epsilon \mb{w} \mb{w}^{T}.
\end{align*}
Then
\begin{enumerate}
\item[(i)] $\tilde {\mb{M}}_1,\tilde {\mb{M}}_2 \succcurlyeq \mb{0}$,
\item[(ii)] $\tilde {\mb{M}}_1+\tilde {\mb{M}}_2 = \bar {\mb{M}}_1 +\bar {\mb{M}}_2={\mb{K}_{\mb{Z}|U}^{*-1}},$
\item[(iii)] \begin{align}
\left (\mb{D_Z} - \mb{K}_{\mb{Z}|U}^{*} \right )\tilde {\mb{M}}_1 &= \left (\mb{D_Z} - \mb{K}_{\mb{Z}|U}^{*} \right ) \left (\bar {\mb{M}}_1 - \epsilon \mb{w} \mb{w}^{T} \right ) \nonumber\\
&=\left (\mb{D_Z} - \mb{K}_{\mb{Z}|U}^{*} \right )\bar {\mb{M}}_1 - \left (\mb{D_Z} - \mb{K}_{\mb{Z}|U}^{*} \right )\epsilon \mb{w} \mb{w}^{T} \nonumber \\
&= \mb{0} - \left (\mb{D_Z} - \mb{K}_{\mb{Z}|U}^{*} \right )\epsilon \sum_{i,j=1}^{r} e_i e_j \mb{s}_i \mb{s}_j^{T} \\
&= \mb{0},
\end{align}
where (25) and (26) follow from (13) and (11), respectively, and
\item[(iv)] \begin{align}
\left (\mb{K_Z} - \mb{K}_{\mb{Z}|U}^{*} \right  )\mb{M}_2^{*} &= \left (\mb{K_Z} - \mb{K}_{\mb{Z}|U}^{*} \right ) \left (\bar {\mb{M}}_2 + \epsilon \mb{w} \mb{w}^{T} \right ) \nonumber\\
&=\left (\mb{K_Z} - \mb{K}_{\mb{Z}|U}^{*} \right )\bar {\mb{M}}_2 + \left (\mb{K_Z} - \mb{K}_{\mb{Z}|U}^{*} \right )\epsilon \mb{w} \mb{w}^{T} \nonumber\\
&= \mb{0},
\end{align}
where (27) follows from (14).
\end{enumerate}
So, $\left ({\mb{K}_{\mb{Z}|U}^{*}}, \tilde {\mb{M}}_1, \tilde {\mb{M}}_2 \right )$ satisfies the KKT conditions, and hence is optimal for the optimization problem (1). But then
\begin{align*}
\textrm{Trace}(\tilde {\mb{M}}_1) &= \textrm{Trace}(\bar {\mb{M}}_1) - \textrm{Trace}(\epsilon \mb{w} \mb{w}^{T}) \\
&< \textrm{Trace}(\bar {\mb{M}}_1),
\end{align*}
which is a contradiction to the assumption that $\left (\bar {\mb{M}}_1, \bar {\mb{M}}_2 \right )$ solves the optimization problem (8). Therefore, $\mb{S}^T \left (\mb{K_Z} - \mb{K}_{\mb{Z}|U}^{*} \right) \mb{S}$ is non-singular.

The proof of part (b) is similar to that of part (a). We first show by contradiction that the columns of $[\mb{S,T}]$ are linearly independent. Suppose otherwise that they are linearly dependent. Note that the columns of $\mb{S}$ and $\mb{T}$ are linearly independent. Therefore, there exists
\begin{align*}
\mb{0} \neq \mb{w} = \sum_{i=1}^{r} a_i \mb{s}_i = \sum_{i=1}^{l} b_i \mb{t}_i,
\end{align*}
where $a_i$'s and $b_i$'s are scalars such that at least one of the $a_i$'s and at least one of the $b_i$'s are nonzero. Pick any $\epsilon$ such that
\begin{align*}
0 < \epsilon \le \frac{\lambda_{min}}{\|\mb{w}\|^2}.
\end{align*}
Then as proved in part (a), $\bar {\mb{M}}_1 - \epsilon \mb{w} \mb{w}^{T}$ is a positive semidefinite matrix and the matrices $\tilde {\mb{M}}_1$ and $\tilde {\mb{M}}_2$ defined as before satisfy (i)-(iii) above. Moreover,
\begin{align}
\left (\mb{K_Z} - \mb{K}_{\mb{Z}|U}^{*} \right  )\mb{M}_2^{*} &= \left (\mb{K_Z} - \mb{K}_{\mb{Z}|U}^{*} \right ) \left (\bar {\mb{M}}_2 + \epsilon \mb{w} \mb{w}^{T} \right ) \nonumber\\
&=\left (\mb{K_Z} - \mb{K}_{\mb{Z}|U}^{*} \right )\bar {\mb{M}}_2 + \left (\mb{K_Z} - \mb{K}_{\mb{Z}|U}^{*} \right )\epsilon \mb{w} \mb{w}^{T} \nonumber\\
&= \mb{0} + \left (\mb{K_Z} - \mb{K}_{\mb{Z}|U}^{*} \right )\epsilon \sum_{i,j=1}^{l} b_i b_j \mb{t}_i \mb{t}_j^{T} \nonumber\\
&= \mb{0},
\end{align}
where (28) follows from (12). So, we again have that $\left ({\mb{K}_{\mb{Z}|U}^{*}}, \tilde {\mb{M}}_1, \tilde {\mb{M}}_2 \right )$ satisfies the KKT conditions, and hence as before we have arrived at a contradiction. Therefore, the columns of $[\mb{S,T}]$ are linearly independent which implies that
\begin{align}
r+l \le m.
\end{align}
Next, we have from (2), (7), (9) and (10) that
\begin{align}
{\mb{K}_{\mb{Z}|U}^{*-1}} &= \bar {\mb{M}}_1 +\bar {\mb{M}}_2
= \sum_{i=1}^{r} \lambda_i \mb{s}_i \mb{s}_i^{T} +\sum_{i=1}^{l} \gamma_i \mb{t}_i \mb{t}_i^{T},
\end{align}
which means that
\begin{align}
r+l \ge m,
\end{align}
because otherwise ${\mb{K}_{\mb{Z}|U}^{*-1}}$ will be singular. (29) and (31) imply that
\begin{align*}
r+l = m,
\end{align*}
which means that $[\mb{S,T}]$ is a square matrix, and is therefore invertible because it has linearly independent columns.

For part (c), on post-multiplying (30) by ${\mb{K}_{\mb{Z}|U}^{*}} \mb{s}_1$, we obtain
\begin{align*}
\mb{s}_1 &= \sum_{i=1}^{r} \lambda_i \mb{s}_i \left (\mb{s}_i^{T} {\mb{K}_{\mb{Z}|U}^{*}} \mb{s}_1 \right )+\sum_{i=1}^{l} \gamma_i \mb{t}_i \left (\mb{t}_i^{T} {\mb{K}_{\mb{Z}|U}^{*}} \mb{s}_1 \right )
\end{align*}
which can be re-written as
\begin{align}
\mb{s}_1 \left (1 - \lambda_1 \left (\mb{s}_1^{T} {\mb{K}_{\mb{Z}|U}^{*}} \mb{s}_1 \right ) \right ) - \sum_{i=2}^{r} \lambda_i \mb{s}_i \left (\mb{s}_i^{T} {\mb{K}_{\mb{Z}|U}^{*}} \mb{s}_1 \right ) &=\sum_{i=1}^{l} \gamma_i \mb{t}_i \left (\mb{t}_i^{T} {\mb{K}_{\mb{Z}|U}^{*}} \mb{s}_1 \right ).
\end{align}
Since the columns of $[\mb{S,T}]$ are linearly independent by part (b), the coefficients of all vectors in (32) must be zero. Therefore,
\begin{align*}
\lambda_1 \mb{s}_1^{T} {\mb{K}_{\mb{Z}|U}^{*}} \mb{s}_1 &= 1 , \\
\mb{s}_i^{T} {\mb{K}_{\mb{Z}|U}^{*}} \mb{s}_1 &= 0, \hspace{0.1in}\forall i \in \{2,\dots,r\}, \\
\mb{t}_i^{T} {\mb{K}_{\mb{Z}|U}^{*}} \mb{s}_1 &= 0, \hspace{0.1in}\forall i \in \{1,\dots,l\}.
\end{align*}
Likewise, on post-multiplying (30) by ${\mb{K}_{\mb{Z}|U}^{*}} \mb{s}_2, \dots, {\mb{K}_{\mb{Z}|U}^{*}} \mb{s}_r, {\mb{K}_{\mb{Z}|U}^{*}} \mb{t}_1 \dots, {\mb{K}_{\mb{Z}|U}^{*}} \mb{t}_l$ and then equating all the coefficients to zero, we obtain similar equations. In summary,
\begin{align*}
\lambda_i \mb{s}_i^{T} {\mb{K}_{\mb{Z}|U}^{*}} \mb{s}_i &= 1, \hspace{0.1in}\forall i \in \{1,\dots,r\},\\
\gamma_i\mb{t}_i^{T} {\mb{K}_{\mb{Z}|U}^{*}} \mb{t}_i &=1 , \hspace{0.1in}\forall i \in \{1,\dots,l\},\\
\mb{s}_i^{T} {\mb{K}_{\mb{Z}|U}^{*}} \mb{s}_j &= 0, \hspace{0.1in}\forall i,j \in \{1,\dots,r\}, i \neq j, \\
\mb{t}_i^{T} {\mb{K}_{\mb{Z}|U}^{*}} \mb{t}_j &= 0, \hspace{0.1in}\forall i,j \in \{1,\dots,l\}, i \neq j, \\
\mb{s}_i^{T} {\mb{K}_{\mb{Z}|U}^{*}} \mb{t}_j &= 0, \hspace{0.1in}\forall i \in \{1,\dots,r\}, \forall j \in \{1,\dots,l\},
\end{align*}
which imply that
\begin{align}
[\mb{S,T}]^T {\mb{K}_{\mb{Z}|U}^{*}} [\mb{S,T}]= \mb{I}_m.
\end{align}
Hence, $[\mb{S,T}]$ is ${\mb{K}_{\mb{Z}|U}^{*}}$-orthogonal.

For parts (d) to (g), we have from (11) and (12) that
\begin{align*}
\mb{D_Z} \mb{S} &= {\mb{K}_{\mb{Z}|U}^{*}} \mb{S}, \\
\mb{K_Z} \mb{T} &= {\mb{K}_{\mb{Z}|U}^{*}} \mb{T},
\end{align*}
which along with (33) imply
\begin{align*}
\mb{S}^T \mb{D_Z} \mb{S} &= \mb{S}^T {\mb{K}_{\mb{Z}|U}^{*}} \mb{S} = \mb{I}_r, \\
\mb{T}^T \mb{K_Z} \mb{T} &= \mb{T}^T {\mb{K}_{\mb{Z}|U}^{*}} \mb{T} = \mb{I}_l, \\
\mb{T}^T \mb{D_Z} \mb{S} &= \mb{T}^T {\mb{K}_{\mb{Z}|U}^{*}} \mb{S} = \mb{0}, \\
\mb{S}^T \mb{K_Z} \mb{T} &= \mb{S}^T {\mb{K}_{\mb{Z}|U}^{*}} \mb{T} = \mb{0}.
\end{align*}
This completes the proof of Theorem 2.
\end{proof}
It is clear from Theorem 2 that span$\{\mb{S}\}$ is the set of directions in which the encoder sends information until the distortion constraint is met with equality. Similarly, span$\{\mb{T}\}$ is the set of directions in which the encoder sends no information and hence $\mb{K}_{\mb{Z}}$ constraint is met with equality in such directions. Note that if $\mb{S}$ is an empty matrix, then the rate-distortion function is zero.
\section{Converse Proof of the Main Result}
An outline of the converse proof is as follows. We start with a single letter outer bound to the rate region $\mathcal{R}(\mb{D})$. The single letter outer bound defines the \emph{main optimization problem} $P$ that lower bounds the first encoder's achievable rate for fixed $\mb{D}$ and $R_2$. We solve the Gaussian version $P_G$ of the \emph{main optimization problem} $P$ by restricting the solution space to Gaussian distributions. We show that the problem $P_G$ can be reduced to a problem similar to (1). Hence, its optimal solution gives two sets of directions $\mb{S}$ and $\mb{T}$ as discussed in Section 4. The distortion constraint is tight in directions spanned by the columns of $\mb{S}$. The idea then is to define a potentially reduced dimensional problem, namely the \emph{reduced main optimization problem} $\tilde P$ by projecting the main source $\mb{X}$ on $\mb{S}$ and by imposing the distortion constraint only in directions spanned by the columns of $\mb{S}$. The \emph{reduced main optimization problem} $\tilde P$ lower bounds the \emph{main optimization problem} $P$ and its optimal solution is Gaussian. Moreover, $P_G$ and ${\tilde P}$ have the same optimal values. Therefore, the optimal solution to the \emph{main optimization problem} $P$ is Gaussian.

Liu and Viswanath gave a single-letter outer bound to the rate region in \cite{Lui}. We arrive at a similar outer bound by using a slightly different outer bounding technique.
\begin{Lem} (Single-letter outer bound)
If the rate-distortion vector $(R_1, R_2, \mb{D})$ is achievable then there exists random variables $U$ and $V$ such that
\begin{align*}
R_1 &\ge I(\mb{X};U|V) \\
R_2 &\ge I(Y;V) \\
\mb{D} &\succcurlyeq \mb{K}_{\mb{X}|U,V} \\
\mb{X} &\leftrightarrow Y \leftrightarrow V.
\end{align*}
\end{Lem}
\begin{proof}
See Appendix B.
\end{proof}

Let us define the \emph{main optimization problem} $P$ as
\begin{align*}
\min_{U,V} \hspace {0.1in} &I(\mb{X};U|V) \nonumber\\
\textrm{subject to} \hspace{0.1in} &R_2 \ge I(Y;V) \\
&\mb{D} \succcurlyeq \mb{K}_{\mb{X}|U,V} \nonumber\\
&\mb{X} \leftrightarrow Y \leftrightarrow V.\nonumber
\end{align*}
We will show that the optimal solution to $P$ is Gaussian. Let us first restrict the solution space to Gaussian distributions. This results in an optimization problem $P_G$ over the conditional covariance matrix $\mb{K}_{\mb{X}|U,V}$ and the conditional variance $\sigma^2_{Y|V}$. Formally, it can be defined as
\begin{align*}
\min_{\mb{K}_{\mb{X}|U,V},\sigma^2_{Y|V}} \hspace {0.1in} &\frac{1}{2} \log \frac{ \left|\mb{K}_{\mb{X}|V} \right |}{ \left |\mb{K}_{\mb{X}|U,V} \right |} \nonumber\\
\textrm{subject to} \hspace{0.1in} &R_2 \ge \frac{1}{2} \log \frac{\sigma^2_Y}{\sigma^2_{Y|V}} \\
&\mb{D} \succcurlyeq \mb{K}_{\mb{X}|U,V} \succcurlyeq \mb{0}\nonumber\\
&\mb{K}_{\mb{X}|V} \succcurlyeq \mb{K}_{\mb{X}|U,V} \nonumber,
\end{align*}
where
\begin{align*}
\mb{K}_{\mb{X}|V} = \mb{a} \mb{a}^{T} \sigma^2_{Y|V} + \mb{K_N}.
\end{align*}
Let us denote the optimal values of $P$ and $P_G$ by $v(P)$ and $v(P_G)$, respectively. The same notation is used to denote the optimal values of other optimization problems defined in the paper. We can rewrite $P_G$ as
\begin{align}
\min_{\sigma^2_{Y|V}} \hspace{0.1in} & \frac{1}{2} \log \left |\mb{K}_{\mb{X}|V} \right| - \frac{1}{2} v \left (F \left (\mb{D}, \mb{K}_{\mb{X}|V} \right ) \right ) \nonumber\\
\textrm{subject to} \hspace{0.1in} &R_2 \ge \frac{1}{2} \log \frac{\sigma^2_Y}{\sigma^2_{Y|V}},
\end{align}
which is a double optimization problem. Note that for a fixed $\sigma^2_{Y|V}$, the inner optimization problem turns out to be $F \left ( \mb{D}, \mb{K}_{\mb{X}|V} \right )$, which was defined in (1).

Since $P_G$ has a continuous objective and a compact feasible set, there exists an optimal solution $\left (\mb{K}_{\mb{X}|U^{*},V^{*}},\sigma^2_{Y|V^{*}}\right )$ to it, where $U^{*}$ and $V^{*}$ represent the corresponding optimal Gaussian random variables. We now have the following lemma which states that it is optimal for the second encoder to use all $R_2$ bits for sending a message to the decoder.
\begin{Lem} There exists an optimal conditional variance $\sigma^2_{Y|V^{*}}$ such that
\begin{align}
\sigma^2_{Y|V^{*}} = \sigma^2_{Y} 2^{-2 R_2}.
\end{align}
\end{Lem}
\begin{proof}
See Appendix C.
\end{proof}
(34) and Lemma 3 immediately imply that the optimal value of $P_G$ is
\begin{align}
v(P_G) = \frac{1}{2} \log \left |\mb{K}_{\mb{X}|V^{*}} \right| - \frac{1}{2} v \left (F \left (\mb{D}, \mb{K}_{\mb{X}|V^{*}} \right ) \right ),
\end{align}
where
\begin{align}
\mb{K}_{\mb{X}|V^{*}} &= \mb{a} \mb{a}^{T} \sigma^2_{Y} 2^{-2 R_2} + \mb{K_N},
\end{align}
and $\mb{K}_{\mb{X}|U^{*},V^{*}}$ is optimal for problem $F \left(\mb{D}, \mb{K}_{\mb{X}|V^{*}} \right )$ with an optimal value
\begin{align}
v \left (F \left (\mb{D}, \mb{K}_{\mb{X}|V^{*}} \right ) \right ) = \frac{1}{2} \log \left |\mb{K}_{\mb{X}|U^{*},V^{*}} \right|.
\end{align}
As discussed in Section 4, $\mb{K}_{\mb{X}|U^{*},V^{*}}$ gives two sets of directions $\mb{S}$ and $\mb{T}$ which satisfy the properties in Theorem 2. On substituting (38) into (36), we obtain
\begin{align}
v(P_G) &= \frac{1}{2} \log \frac{\left |\mb{K}_{\mb{X}|V^{*}} \right |}{ \left |\mb{K}_{\mb{X}|U^{*},V^{*}} \right |} \nonumber\\
&= \frac{1}{2} \log \frac{\left |[\mb{S,T}]^T \mb{K}_{\mb{X}|V^{*}} [\mb{S,T}] \right |}{ \left | [\mb{S,T}]^T \mb{K}_{\mb{X}|U^{*},V^{*}} [\mb{S,T}]\right |} \\
&= \frac{1}{2} \log \frac{\left |\mat{\mb{S}^T \mb{K}_{\mb{X}|V^{*}} \mb{S} & \mb{S}^T \mb{K}_{\mb{X}|V^{*}} \mb{T} \\ \mb{T}^T \mb{K}_{\mb{X}|V^{*}} \mb{S} & \mb{T}^T \mb{K}_{\mb{X}|V^{*}} \mb{T}} \right |}{ \left | \mb{I}_m \right |} \\
&= \frac{1}{2} \log \left |\mat{\mb{S}^T \mb{K}_{\mb{X}|V^{*}} \mb{S} & \mb{0} \\ \mb{0} & \mb{I}_l} \right | \\
&= \frac{1}{2} \log \left |\mb{S}^T \mb{K}_{\mb{X}|V^{*}} \mb{S} \right|,
\end{align}
where
\begin{enumerate}
\item[(39)] follows because $[\mb{S,T}]$ is invertible from Theorem 2(b),
\item[(40)] follows because $[\mb{S,T}]$ is $\mb{K}_{\mb{X}|U^{*},V^{*}}$-orthogonal from Theorem 2(c), and
\item[(41)] follows because $\mb{T}$ is $\mb{K}_{\mb{X}|V^{*}}$-orthogonal, and $\mb{S}$ and $\mb{T}$ are cross $\mb{K}_{\mb{X}|V^{*}}$-orthogonal from Theorem 2(e) and 2(g), respectively.
\end{enumerate}
We now have the following theorem which is central to the converse proof of our main result.
\begin{Thm}
A Gaussian $(U,V)$ is an optimal solution of the main optimization problem $P$.
\end{Thm}
\begin{proof}
First note that since restricting the solution space over Gaussian distributions can only increase the optimal value of the \emph{main optimization problem} $P$, we immediately have
\begin{align*}
v(P_G) \ge v(P).
\end{align*}
So, it suffices to prove the reverse inequality
\begin{align*}
v(P_G) \le v(P).
\end{align*}
Let us define the \emph{reduced main optimization problem} $\tilde P$ as
\begin{align*}
\min_{U,V} \hspace {0.1in} &I \left (\mb{S}^T\mb{X};U|V \right) \nonumber\\
\textrm{subject to} \hspace{0.1in} &R_2 \ge I(Y;V) \\
&\mb{S}^T\mb{D S} \succcurlyeq \mb{S}^T\mb{K}_{\mb{X}|U,V} \mb{S} \nonumber\\
&\mb{S}^T\mb{X} \leftrightarrow Y \leftrightarrow V.\nonumber
\end{align*}
We will show that the \emph{main optimization problem} $P$ is lower bounded by the \emph{reduced main optimization problem} $\tilde P$. Since $[\mb{S,T}]$ is invertible from Theorem 2(b) and the mutual information is non-negative, we have
\begin{align}
I(\mb{X};U|V) &= I\left ( \left [\mb{S,T} \right]^T \mb{X};U|V \right ) \nonumber\\
&=I \left(\mb{S}^T \mb{X}, \mb{T}^T \mb{X};U|V \right) \nonumber\\
&= I \left (\mb{S}^T \mb{X};U|V \right) + I \left(\mb{T}^T \mb{X};U|V,\mb{S}^T \mb{X} \right) \nonumber\\
&\ge I \left(\mb{S}^T \mb{X};U|V \right).
\end{align}
Note that any $(U,V)$ satisfying
\begin{align*}
&\mb{D} \succcurlyeq \mb{K}_{\mb{X}|U,V} \\
&\mb{X} \leftrightarrow Y \leftrightarrow V
\end{align*}
also satisfies
\begin{align*}
&\mb{S}^T \mb{D}\mb{S} \succcurlyeq \mb{S}^T\mb{K}_{\mb{X}|U,V} \mb{S} \\
&\mb{S}^T \mb{X} \leftrightarrow Y \leftrightarrow V.
\end{align*}
Therefore, the feasible set of $P$ is contained in that of $\tilde P$. Moreover, (43) above implies that the objective of $P$ is no less than that of $\tilde P$. We therefore have that the \emph{reduced main optimization problem} $\tilde P$ lower bounds the \emph{main optimization problem }$P$, i.e.
\begin{align}
v(P) &\ge v(\tilde P).
\end{align}
The objective of $\tilde P$ can be decomposed as
\begin{align}
I \left (\mb{S}^T\mb{X};U|V\right ) = I \left (\mb{S}^T\mb{X};U,V\right ) - I \left (\mb{S}^T\mb{X};V\right ).
\end{align}
We now define two subproblems that are used to lower bound the \emph{reduced main optimization problem} $\tilde P$. The first subproblem $\tilde P(\mb{D})$ minimizes the first mutual information in the right-hand-side of (45) subject to the distortion constraint in $\tilde P$ and the second subproblem $\tilde P(R_2)$ maximizes the second mutual information in the right-hand-side of (45) subject to the rate constraint and the Markov condition in $\tilde P$. In other words, $\tilde P(\mb{D})$ is defined as
\begin{align*}
\min_{U,V} \hspace {0.1in} &I \left(\mb{S}^T\mb{X};U,V \right) \\
\textrm{subject to} \hspace{0.1in} &\mb{S}^T\mb{D S} \succcurlyeq \mb{S}^T\mb{K}_{\mb{X}|U,V} \mb{S},
\end{align*}
and $\tilde P(R_2)$ is defined as
\begin{align*}
\max_{V} \hspace {0.1in} &I \left (\mb{S}^T\mb{X};V \right) \\
\textrm{subject to} \hspace{0.1in} &R_2 \ge I(Y;V) \\
&\mb{S}^T\mb{X} \leftrightarrow Y \leftrightarrow V.
\end{align*}
It is clear from the decomposition in (45) and from the definitions of $\tilde P, \tilde P(\mb{D})$ and $\tilde P(R_2)$ that $\tilde P(\mb{D})$ and $\tilde P(R_2)$ lower bound $\tilde P$, i.e.
\begin{align}
v(\tilde P) &\ge v\left (\tilde P(\mb{D}) \right ) - v\left (\tilde P(R_2)\right ).
\end{align}
We now give two Lemmas about the optimal solutions to subproblems $\tilde P(\mb{D})$ and $\tilde P(R_2)$.
\begin{Lem} A Gaussian $(U,V)$ with the conditional covariance matrix $\mb{K}_{\mb{X}|U^{*},V^{*}}$ is optimal for the subproblem $\tilde P(\mb{D})$, and the optimal value is
\begin{align}
v\left(\tilde P(\mb{D}) \right ) &= \frac{1}{2} \log \left | \mb{S}^T\mb{K_X} \mb{S} \right |.
\end{align}
\end{Lem}
\begin{proof}
See Appendix D.
\end{proof}
\begin{Lem} A Gaussian $V$ with the conditional variance $\sigma^2_{Y|V^{*}}$ is optimal for the subproblem $\tilde P(R_2)$, and the optimal value is
\begin{align}
v\left(\tilde P(R_2) \right) &= \frac{1}{2} \log \frac {\left | \mb{S}^T\mb{K_X} \mb{S} \right |}{\left | \mb{S}^T\mb{K}_{\mb{X}|V^{*}} \mb{S} \right |}.
\end{align}
\end{Lem}
\begin{proof}
See Appendix E.
\end{proof}
Substituting (47) and (48) into (46), we get
\begin{align}
v(\tilde P) &\ge \frac{1}{2} \log \left | \mb{S}^T\mb{K}_{\mb{X}|V^{*}} \mb{S} \right |.
\end{align}
From (42), (44) and (49), we have
\begin{align*}
v(P) &\ge v(P_G),
\end{align*}
which means that a Gaussian $(U,V)$ is optimal for the \emph{main optimization problem} $P$.
\end{proof}
We are now ready to prove the converse of Theorem 1. Suppose $(R_1,R_2, \mb{D})$ is achievable, then
\begin{align}
R_1 &\ge v(P) \\
&=v(P_G) \\
&= \frac{1}{2} \log \left |\mb{K}_{\mb{X}|V^{*}} \right | - \frac{1}{2} v \left (F \left (\mb{D}, \mb{K}_{\mb{X}|V^{*}} \right ) \right)\\
&= \min_{\mb{K}} \hspace {0.1in} \frac{1}{2} \log \frac{ \left |\mb{a} \mb{a}^{T} \sigma_{Y}^2 2^{-2 R_2} + \mb{K_{N}} \right |}{|\mb{K}|} \nonumber\\
&\hspace{0.25in}\textrm{s. t.} \hspace{0.1in} \mb{0} \preccurlyeq \mb{K} \preccurlyeq \mb{D}\\
&\hspace{0.55in}\mb{K} \preccurlyeq \mb{a} \mb{a}^{T} \sigma_{Y}^2 2^{-2 R_2} + \mb{K_N}\nonumber,
\end{align}
where
\begin{enumerate}
\item[(50)] follows from Lemma 2,
\item[(51)] follows from Theorem 3,
\item[(52)] follows from (36), and
\item[(53)] follows from the definition of $F$ and (37).
\end{enumerate}
And if $(R_1,R_2,\mb{D}) \in \hat {\mathcal{R}}$, then (53) again holds because (52) is continuous in $(R_2, \mb{D})$. This completes the converse proof of Theorem 1.

\section{The General Source-Coding Problem}
In this section, we study the general source-coding problem by imposing separate distortion constraints on both the sources. The mathematical formulation of the problem remains the same as in Section 3.1. The only change is in the decoder which now uses the received messages from the encoders to estimate both $\mb{X}^n$ and $Y^n$ using the decoding functions
\begin{align*}
g_1^{(n)}  &: \left \{1,\dots,M_1^{(n)} \right \} \times \left \{1,\dots,M_2^{(n)} \right \} \mapsto \mathbb{R}^{mn}, \hspace{0.1in} \textrm{and}\\
g_2^{(n)}  &: \left \{1,\dots,M_1^{(n)} \right \} \times \left \{1,\dots,M_2^{(n)} \right \} \mapsto \mathbb{R}^{n},
\end{align*}
respectively.
\begin{Def}
A rate-distortion vector $\left (R_1,R_2, \mb{D}, d \right )$, where $\mb{D}$ is a positive definite matrix, is \emph{achievable} for the general source-coding problem if there exists a block length $n$, encoders $f_1^{(n)} $ and $f_2^{(n)} $, and a decoder $\left (g_1^{(n)}, g_2^{(n)}\right ) $ such that
\begin{align*}
R_i &\ge \frac{1}{n} \log M_i^{(n)}  \hspace {0.15 cm} \textrm{for all} \hspace {0.15 cm} i \in \{1,2\}, \\
\mb{D} &\succcurlyeq \frac{1}{n} \sum_{i=1}^n E \left [ \left (\mb{X}_i - \hat {\mb{X}}_i \right ) \left (\mb{X}_i - \hat {\mb{X}}_i \right )^T \right ], \hspace {0.15 cm} \textrm{and} \\
d &\ge \frac{1}{n} \sum_{i=1}^n E \left [ \left (Y_i - \hat Y_i \right )^2 \right ],
\end{align*}
where
\begin{align*}
\hat {\mb{X}}^n &= g_1^{(n)} \left (f_1^{(n)} \left (\mb{X}^n \right ), f_2^{(n)}\left (Y^n \right ) \right ) \\
&= E \left [\mb{X}^n | f_1^{(n)} \left (\mb{X}^n \right ), f_2^{(n)}\left (Y^n \right )\right ],  \hspace {0.15 cm} \textrm{and} \\
\hat Y^n &= g_2^{(n)} \left (f_1^{(n)} \left (\mb{X}^n \right ), f_2^{(n)}\left (Y^n \right ) \right ) \\
&= E \left [Y^n | f_1^{(n)} \left (\mb{X}^n \right ), f_2^{(n)}\left (Y^n \right )\right ].
\end{align*}
Let $\tilde {\mathcal{R}}$ be the closure of the set of all achievable rate-distortion vectors. Define
\[
\mathcal{R} \left (\mb{D},d \right) = \left \{ (R_1,R_2) : \left (R_1,R_2, \mb{D}, d \right ) \in \tilde {\mathcal{R}} \right\}.
\]
We call $\mathcal{R} \left (\mb{D},d \right)$ the rate region for the general source-coding problem.
\end{Def}

We can assume without loss of generality that the components $(X_1,\dots,X_m)$ of $\mb{X}$ and $Y$ are standard normal, and $(X_1,Y)$ is independent of $(X_2,\dots,X_m)$. Starting from any problem, we can get to an equivalent problem with this structure by applying an invertible transformation on the sources and by considering the equivalent distortion constraints \cite{Chao, Globerson}. This can be done as follows. Let
\[
\mb{u}_1, \mb{u}_2, \dots, \mb{u}_m
\]
be an orthonormal basis in $\mathbb{R}^m$ starting at
\begin{align*}
\mb{u}_1 =  \frac{1}{\rho} \left (\sigma_Y\mb{K}_{\mb{X}}^{-1/2} \mb{a} \right ),
\end{align*}
where
\[
\rho = \left \|\sigma_Y\mb{K}_{\mb{X}}^{-1/2} \mb{a} \right \|.
\]
Define the matrices
\begin{align*}
\mb{U} &= \left [\mb{u}_1, \mb{u}_2, \dots, \mb{u}_m \right ]\\
\mb{T}_{\mb{X}} &= \mb{U}^T \mb{K}_{\mb{X}}^{-1/2}.
\end{align*}
Then the transformation is given by
\begin{align*}
\tilde {\mb{X}} &= \mb{T}_{\mb{X}} \mb{X}\\
\tilde Y &= \frac{1}{\sigma_Y} Y.
\end{align*}
The covariance matrix of $\tilde {\mb{X}}$ is
\begin{align*}
\mb{K}_{\tilde {\mb{X}}} &=  \mb{T}_{\mb{X}} \mb{K_X} \mb{T}_{\mb{X}}^T\\
&= \mb{U}^T \mb{K}_{\mb{X}}^{-1/2}  \mb{K_X} \mb{K}_{\mb{X}}^{-1/2} \mb{U}\\
&= \mb{U}^T \mb{U} \\
&= \mb{I}_m,
\end{align*}
and the cross-covariance between $\tilde {\mb{X}}$ and $\tilde Y$ is
\begin{align*}
\mb{K}_{\tilde {\mb{X}} \tilde Y} &=\frac{1}{\sigma_Y}\mb{T}_{\mb{X}} \mb{K}_{{\mb{X}}Y} \\
&=\frac{1}{\sigma_Y} \left (\mb{U}^T \mb{K}_{\mb{X}}^{-1/2} \right) \left (\sigma^2_Y \mb{a} \right) \\
&= \mb{U}^T \left (\sigma_Y \mb{K}_{\mb{X}}^{-1/2} \mb{a} \right)\\
&= \mb{U}^T \left (\rho \mb{u}_1\right)  \\
&=\left (\rho ,0, \dots, 0\right)^T.
\end{align*}
Now, it is easy to verify that the equivalent distortion constraints are
\begin{align*}
\mb{T}_{\mb{X}} \mb{D} \mb{T}_{\mb{X}}^T &\succcurlyeq \frac{1}{n} \sum_{i=1}^n E \left [ \left (\tilde {\mb{X}}_i - \hat {\tilde {\mb{X}}}_i \right ) \left (\tilde {\mb{X}}_i - \hat {\tilde {\mb{X}}}_i \right )^T \right ], \hspace {0.15 cm} \textrm{and} \\
\frac{d}{\sigma^2_Y} &\ge \frac{1}{n} \sum_{i=1}^n E \left [ \left (\tilde Y_i - \hat {\tilde Y}_i \right )^2 \right ].
\end{align*}
Since the above transformation is invertible, it does not incur any information loss. We therefore have an equivalent structured problem. So from now on, we will assume that our original problem has this structure with $\rho$ being the correlation coefficient between $X_1$ and $Y$. 

\subsection{An Outer Bound}
First note that if there is no distortion constraint between $Y^n$ and $\hat Y^n$, then the problem reduces to the Gaussian scalar-help-vector source-coding problem, and hence we have
\begin{align}
\mathcal{R}(\mb{D},d) \subseteq \mathcal{R}^{*}(\mb{D}).
\end{align}
This bound is tight for large $R_2$ because the distortion constraint between $Y^n$ and $\hat Y^n$ is always satisfied for large $R_2$. It prompts us to consider another relaxed problem in which there is no distortion constraint between $\mb{X}^n$ and $\hat {\mb{X}}^n$, and obtain another outer bound to the rate region. However, the outer bound thus obtained is not tight in general even for large $R_1$. This is because the optimal solution to this relaxed problem is such that the first encoder sends information about $X_1$ only. The rest of the components of $\mb{X}$ are simply ignored and therefore the distortion constraint between $\mb{X}^n$ and $\hat {\mb{X}}^n$ is not met in general. We can obtain an improved outer bound to the rate region by splitting the first encoder's rate into two. The first split of the rate comes from the rate region of the quadratic Gaussian two-encoder source-coding problem for the sources $X_1$ and $Y$ under appropriate individual distortion constraints (\cite{Oohama},\cite{Wagner}). The second split is the point-to-point rate-distortion function for the vector source $(X_2,\dots,X_m)$ under an appropriate covariance matrix distortion constraint. By combining this outer bound with that in (54), we obtain a composite outer bound to the rate region.

Let us denote $(X_2,\dots,X_m)^T$ by $\bar {\mb{X}}$ and let
\[
\mb{D} = \mat{D_1 & \mb{b}^T \\ \mb{b} & \bar {\mb{D}}},
\]
where $D_1$ is a positive number, $\mb{b}$ is a $(m-1)$-dimensional vector and $\bar {\mb{D}}$ is a $(m-1) \times (m-1)$ positive definite matrix. Let $\bar R_1(\bar {\mb{D}})$ be the point-to-point rate-distortion function of the source $\bar {\mb{X}}$ under a covariance matrix distortion constraint $\bar {\mb{D}}$. Then from the discussion in Section 4, we have
\begin{align}
\bar R_1(\bar {\mb{D}}) &=-\frac{1}{2} v\left ( F \left (\bar {\mb{D}},\mb{I}_{m-1} \right )\right ) \nonumber \\
&= -\frac{1}{2}\log | \bar {\mb{D}}^{*} |,
\end{align}
where $\bar{\mb{D}}^{*}$ is the optimal solution to problem $F \left (\bar {\mb{D}},\mb{I}_{m-1} \right )$ defined in (1). Define the sets
\begin{align*}
\mathcal{R}^{*}_2(d) &= \left \{(R_1,R_2) : R_2 \ge \frac{1}{2} \log^{+} \left [\frac{1}{d} \left ( 1- \rho^2 + \rho^2 2^{- 2 \left (R_1-\bar R_1(\bar {\mb{D}})\right)} \right ) \right ]\right \}, \\
\mathcal{R}^{*}_{\textrm{sum}}(D_1,d) &= \left \{(R_1,R_2) : R_1-\bar R_1(\bar {\mb{D}})+R_2 \ge \frac{1}{2} \log^{+} \left [\frac{\left ( 1- \rho^2 \right ) \beta (D_1 , d)}{2 D_1 d} \right ]\right \},
\end{align*}
where $\log^{+} x = \max (\log x, 0)$, and
\[
\beta (D_1 , d) = 1 + \sqrt{1+\frac{4 \rho^2 D_1 d}{\left(1-\rho^2 \right)^2}}.
\]
We now have the following outer bound to the rate region of the general source-coding problem.
\begin{Thm}
For every positive definite matrix $\mb{D}$ and for every positive number $d$
\begin{align}
\mathcal{R}(\mb{D},d) \subseteq \mathcal{R}^{*}(\mb{D}) \cap \mathcal{R}^{*}_2(d) \cap \mathcal{R}^{*}_{\textrm{\emph{sum}}}(D_1,d).
\end{align}
\end{Thm}
\begin{proof}
Consider $(R_1,R_2) \in \mathcal{R}(\mb{D},d)$. Let $C_1 = f_1^{(n)} \left (\mb{X}^n \right )$ and $C_2 =f_2^{(n)} (Y^n).$
Then
\begin{align}
n R_2 &\ge \log M_2^{(n)} \nonumber\\
&\ge H(C_2) \nonumber\\
&\ge H(C_2|C_1) \nonumber\\
&\ge I(Y^n;C_2|C_1) \nonumber\\
&= I(Y^n;C_1,C_2) - I(Y^n;C_1) \nonumber\\
&\ge I(Y^n;\hat Y^n) - I(Y^n;C_1), \\
n R_1 &\ge \log M_1^{(n)} \nonumber\\
&\ge H(C_1) \nonumber\\
&\ge I(\bar {\mb{X}}^n;C_1) \nonumber\\
&= I(X_1^n;C_1)+I(\bar {\mb{X}}^n;C_1|X_1^n).
\end{align}
We can lower bound the second mutual information in (58) as follows
\begin{align}
I(\bar {\mb{X}}^n;C_1|X_1^n)&= I(\bar {\mb{X}}^n;C_1,X_1^n)\nonumber\\
&= I \left(\bar {\mb{X}}^n;C_1,X_1^n,Y^n \right) \\
&= I \left(\bar {\mb{X}}^n;C_1,C_2,X_1^n,Y^n \right) \nonumber\\
&\ge I \left (\bar {\mb{X}}^n;C_1,C_2 \right) \nonumber\\
&\ge I\left(\bar {\mb{X}}^n;\hat {\bar {\mb{X}}}^n\right),
\end{align}
where (59) follows because
\[
\bar {\mb{X}}^n \leftrightarrow (C_1,X_1^n) \leftrightarrow Y^n.
\]
Define the following optimization problem
\begin{align}
\min_{C_1} \hspace {0.1in} &\frac{1}{n}I \left(\bar {\mb{X}}^n;\hat {\bar {\mb{X}}}^n \right) \nonumber\\
\textrm{subject to} \hspace{0.1in} &\frac{1}{n} \sum_{i=1}^n E \left [ \left (\bar {\mb{X}}_i - \hat {\bar {\mb{X}}}_i \right ) \left (\bar{\mb{X}}_i - \hat {\bar {\mb{X}}}_i \right )^T \right ] \preccurlyeq \bar {\mb{D}}.
\end{align}
This is the point-to-point rate-distortion problem for the source $\bar {\mb{X}}^n$ under a covariance matrix distortion constraint $\bar {\mb{D}}$. Therefore from the discussion in Section 4, a Gaussian $C_1$ is optimal for this problem and from (55), the optimal conditional covariance matrix is $\bar {\mb{D}}^{*}$ and the optimal value $\bar R_1(\bar {\mb{D}})$. From (57), (58), (60) and the definition of the optimization problem (61), we obtain that $(R_1,R_2)$ satisfies
\begin{align*}
R_2 &\ge \frac{1}{n}I(Y^n;\hat Y^n) - \frac{1}{n}I(Y^n;C_1) \\
R_1-\bar R_1(\bar {\mb{D}}) &\ge \frac{1}{n}I(X_1^n;C_1) \\
d &\ge \frac{1}{n} \sum_{i=1}^n E \left [ \left (Y_i - \hat Y_i \right )^2 \right ].
\end{align*}
By invoking Oohama's lower bounding technique \cite{Oohama} next, we obtain
\begin{align*}
R_2 \ge \frac{1}{2} \log^{+} \left [\frac{1}{d} \left ( 1- \rho^2 + \rho^2 2^{- 2 (R_1-\bar R_1(\bar {\mb{D}}))} \right ) \right ],
\end{align*}
which implies that
\begin{align}
(R_1, R_2) \in \mathcal{R}^{*}_2(d).
\end{align}
We now proceed to lower bound the sum-rate.
\begin{align}
n (R_1+R_2) &\ge H(C_1,C_2) \nonumber\\
&\ge I(\mb{X}^n,Y^n;C_1,C_2) \nonumber\\
&= I(X_1^n,Y^n;C_1,C_2)+I(\bar {\mb{X}}^n;C_1,C_2|X_1^n,Y^n) \nonumber\\
&= I(X_1^n,Y^n;C_1,C_2)+I(\bar {\mb{X}}^n;C_1|X_1^n) \nonumber\\
&= I(X_1^n,Y^n;C_1,C_2)+I(\bar {\mb{X}}^n;\hat {\bar {\mb{X}}}^n),
\end{align}
where (63) follows from (60). The sum-rate can be lower bounded further by minimizing two mutual informations in (63) separately subject to separate distortion constraints. Using the sum-rate lower bounding technique by Wagner \emph{et al.} \cite{Wagner}, the first mutual information is minimized subject to the distortion constraints $D_1$ and $d$ on the sources $X_1^n$ and $Y^n$, respectively. We omit the details to avoid repetition. Minimizing the second mutual information subject to the covariance matrix distortion constraint $\bar {\mb{D}}$ is the optimization problem (61) again.  We therefore conclude that
\begin{align*}
R_1-\bar R_1(\bar {\mb{D}})+R_2 \ge \frac{1}{2} \log^{+} \left [\frac{\left ( 1- \rho^2 \right ) \beta (D_1 , d)}{2 D_1 d} \right ],
\end{align*}
which implies that
\begin{align}
(R_1, R_2) \in \mathcal{R}^{*}_{\textrm{sum}}(D_1,d).
\end{align}
(54), (62) and (64) together establish the outer bound (56). This completes the proof of Theorem 4.
\end{proof}
\subsection{Tightness of the Outer Bound}
We will prove that the boundary of the rate region $\mathcal{R}(\mb{D},d)$ partially coincides with the boundary of $\mathcal{R}^{*}(\mb{D})$ in general, coincides with the outer bound (56) completely if $\mb{b}=\mb{0}$, and partially coincides with the boundary of $\mathcal{R}^{*}_2(d)$ if $\mb{b}\neq\mb{0}$ and a condition holds. Let
\[
R_1^{*} = \inf \left\{R_1 : R_1 > \bar R_1(\bar {\mb{D}}) + \frac{1}{2} \log \frac{1}{D_1} \hspace{0.05in} \textrm{and} \hspace{0.05in}\bar {\mb{D}} - \bar {\mb{D}}^{*} \succcurlyeq  \frac{\mb{b} \mb{b}^T}{D_1 - 2^{-2 \left(R_1 - \bar R_1(\bar {\mb{D}}) \right)}}  \right\}.
\]
We have the following lemma.
\begin{Lem}
\begin{enumerate}
\item[(a)] There exists a positive number $R_2^{*}$ such that
\begin{align}
\mathcal{R}(\mb{D},d) \cap \left \{R_2 \ge R_2^{*} \right\} = \mathcal{R}^{*}(\mb{D}) \cap \left \{R_2 \ge R_2^{*} \right\}.
\end{align}
\item[(b)] If $\mb{b} = \mb{0}$, then
\begin{align}
\mathcal{R}(\mb{D},d) = \mathcal{R}^{*}(\mb{D}) \cap \mathcal{R}^{*}_2(d) \cap \mathcal{R}^{*}_{\textrm{\emph{sum}}}(D_1,d).
\end{align}
\item[(c)] If $\mb{b} \neq \mb{0}$ and $R_1^{*} < \infty$, then
\begin{align}
\mathcal{R}(\mb{D},d) \cap \left \{R_1 \ge R_1^{*} \right\} = \mathcal{R}^{*}_2(d) \cap \left \{R_1 \ge R_1^{*} \right\}.
\end{align}
\end{enumerate}
\end{Lem}
\begin{proof}
As explained in Section 3.3, the optimal scheme depicted in Fig. 2 is such that the second encoder vector quantizes its observation using a Gaussian test channel as in point-to-point rate-distortion theory. So, the average distortion between $Y^n$ and $\hat {Y}^n$ decreases as $R_2$ increases. Hence, there exists a positive number $R_2^{*}$ such that for any $R_2 \ge R_2^{*}$, the distortion constraint between $Y^n$ and $\hat {Y}^n$ is satisfied and therefore the region
\[
\mathcal{R}^{*}(\mb{D}) \cap \left \{R_2 \ge R_2^{*} \right\}
\]
is achievable for the general source coding problem. This along with the outer bound (54) prove the equality in (65).

For part (b), we will first prove that if $\mb{b} = \mb{0}$, i.e., $\mb{D}$ is block diagonal, then
\begin{align}
\mathcal{R}^{*}(\mb{D}) = \left \{(R_1,R_2) : R_1-\bar R_1(\bar {\mb{D}}) \ge \frac{1}{2} \log^{+} \left [\frac{1}{D_1} \left ( 1- \rho^2 + \rho^2 2^{- 2 R_2} \right ) \right ]\right \}.
\end{align}
The optimization problem in the definition of $\mathcal{R}^{*}(\mb{D})$ is
\begin{align}
\min_{\mb{K}} \hspace {0.1in} &\frac{1}{2} \log \frac{ \left |\mb{a} \mb{a}^{T} \sigma_{Y}^2 2^{-2 R_2} + \mb{K_{N}} \right |}{ \left|\mb{K} \right|} \nonumber\\
\textrm{subject to} \hspace{0.1in} &\mb{0} \preccurlyeq \mb{K} \preccurlyeq \mb{D}\\
&\mb{K} \preccurlyeq \mb{a} \mb{a}^{T} \sigma_{Y}^2 2^{-2 R_2} + \mb{K_N}.\nonumber
\end{align}
Since our problem has a special structure explained above, we have
\[
\mb{a} \mb{a}^{T} \sigma_{Y}^2 2^{-2 R_2} + \mb{K_N} = \textrm{Diag}\{1- \rho^2 + \rho^2 2^{- 2 R_2}, 1, \dots, 1\},
\]
which is an $m \times m$ diagonal matrix. Now, consider any feasible
\[
\mb{K} = \mat{K_1 & \mb{c}^T \\ \mb{c} & \bar {\mb{K}}}.
\]
where $K_1$ is a positive number, $\mb{c}$ is a $(m-1)$-dimensional non-zero vector and $\bar {\mb{K}}$ is a $(m-1) \times (m-1)$ positive definite matrix. Let us define
\[
\tilde {\mb{K}} = \mat{K_1 & \mb{0} \\ \mb{0} & \bar {\mb{K}}}.
\]
Then,
\begin{align*}
|\mb{K}| &= |\bar {\mb{K}}| (K_1 - \mb{c}^T \bar {\mb{K}}^{-1} \mb{c}) \\
&< |\bar {\mb{K}}| K_1 \\
&= |\tilde {\mb{K}}|.
\end{align*}
Therefore, without loss of optimality, we can restrict the feasible solutions to be of the following form
\[
\tilde {\mb{K}} = \mat{K_1 & \mb{0} \\ \mb{0} & \bar {\mb{K}}}.
\]
The restricted feasible set
\[
\left \{\mb{0} \preccurlyeq \tilde {\mb{K}} \preccurlyeq \mb{D} \hspace{0.1in} \textrm{and} \hspace{0.1in} \tilde {\mb{K}} \preccurlyeq \textrm{Diag}\{1- \rho^2 + \rho^2 2^{- 2 R_2}, 1, \dots, 1\} \right \}
\]
is equivalent to
\[
\left \{0 \le K_1 \le \min\left(D_1, 1- \rho^2 + \rho^2 2^{- 2 R_2}\right), \hspace{0.1in} \mb{0} \preccurlyeq \bar {\mb{K}} \preccurlyeq \bar {\mb{D}} \hspace{0.1in} \textrm{and} \hspace{0.1in} \bar {\mb{K}} \preccurlyeq  \mb{I}_{m-1} \right \}.
\]
Now, the objective of the optimization problem (69) can be re-written as
\begin{align*}
\frac{1}{2} \log \frac{ 1- \rho^2 + \rho^2 2^{- 2 R_2}}{ | \tilde {\mb{K}} |} &= \frac{1}{2} \log \frac{ 1- \rho^2 + \rho^2 2^{- 2 R_2}}{ | \bar {\mb{K}} | K_1} \\
&=\frac{1}{2} \log \frac{ 1- \rho^2 + \rho^2 2^{- 2 R_2}}{ K_1} + \frac{1}{2} \log \frac{1}{\bar {\mb{K}}}.
\end{align*}
Therefore, the optimal value of problem (69) equals the sum of the optimal values of subproblems
\begin{align*}
\min_{\bar {\mb{K}}} \hspace {0.1in} &\frac{1}{2} \log \frac{1}{|\bar {\mb{K}}|} \nonumber\\
\textrm{subject to} \hspace{0.1in} &\mb{0} \preccurlyeq \bar {\mb{K}} \preccurlyeq \bar {\mb{D}}\\
&\bar {\mb{K}} \preccurlyeq \mb{I}_{m-1}\nonumber,
\end{align*}
and
\begin{align*}
\min_{K_1} \hspace {0.1in} &\frac{1}{2} \log \frac{1- \rho^2 + \rho^2 2^{- 2 R_2}}{K_1} \nonumber\\
\textrm{subject to} \hspace{0.1in} &0 \le K_1 \le \min\left(D_1, 1- \rho^2 + \rho^2 2^{- 2 R_2}\right).
\end{align*}
The first subproblem is the point-to-point rate-distortion problem for the source $\bar {\mb{X}}$ under a distortion constraint $\bar {\mb{D}}$, and hence its optimal value is $\bar R_1(\bar {\mb{D}})$. The optimal value of the second subproblem is
\[
\frac{1}{2} \log^{+} \left [\frac{1}{D_1} \left ( 1- \rho^2 + \rho^2 2^{- 2 R_2} \right ) \right ].
\]
Thus, the optimal value of the optimization problem (69) is
\[
\bar R_1(\bar {\mb{D}}) + \frac{1}{2} \log^{+} \left [\frac{1}{D_1} \left ( 1- \rho^2 + \rho^2 2^{- 2 R_2} \right ) \right ],
\]
which proves the equality in (68). It is now easy to verify that the outer bound (56) coincides with the shifted boundary of the rate region of the quadratic Gaussian two-encoder source-coding problem \cite{Wagner}, where the shift is by the amount $\bar R_1(\bar {\mb{D}})$ in the direction of $R_1$ axis. So, by using the point-to-point rate-distortion optimal code for the source $\bar {\mb{X}}$ in conjunction with the separation-based optimal scheme for the sources $X_1$ and $Y$ \cite{Wagner}, we can achieve the outer bound. We therefore have the equality in (66).

For part (c), it suffices to show that if the conditions in Lemma 6(c) hold, then
\begin{align}
\mat{2^{-2 \left(R_1^{*} - \bar R_1(\bar {\mb{D}}) \right)} & \mb{0} \\ \mb{0} & \bar {\mb{D}}^{*}} \preccurlyeq \mat{D_1 & \mb{b}^T \\ \mb{b} & \bar {\mb{D}}} =  \mb{D}.
\end{align}
This will imply that the region
\[
\mathcal{R}^{*}_2(d) \cap \left \{R_1 \ge R_1^{*} \right\}
\]
is achievable for the general source-coding problem by using a scheme in which the source $\bar {\mb{X}}$ is encoded and decoded as in the point-to-point rate-distortion theory, and the sources $X_1$ and $Y$ are encoded and decoded as in the scalar Gaussian one-helper problem \cite{Oohama}, treating $Y$ as the main source and $X_1$ as the helper. Consider any
\[
\mb{0} \neq \mb{x} = \mat{y \\ \mb{z}} \in \mathbb{R}^m,
\]
where $y$ is a scalar and $\mb{z}$ is a $(m-1)$-dimensional vector. Then
\begin{align}
&\mb{x}^T \left[\mb{D}-\mat{2^{-2 \left(R_1^{*} - \bar R_1(\bar {\mb{D}}) \right)} & \mb{0} \\ \mb{0} & \bar {\mb{D}}^{*}}\right] \mb{x} \nonumber\\
= \hspace{0.05in}&\mat{y & \mb{z}^T} \mat{D_1-2^{-2 \left(R_1^{*} - \bar R_1(\bar {\mb{D}}) \right)} & \mb{b}^T \\ \mb{b} & \bar {\mb{D}}-\bar {\mb{D}}^{*}} \mat{y \\ \mb{z}} \nonumber\\
= \hspace{0.05in}& y^2 \left (D_1-2^{-2 \left(R_1^{*} - \bar R_1(\bar {\mb{D}}) \right)} \right) + 2 y (\mb{z}^T \mb{b})+ \mb{z}^T \left ( \bar {\mb{D}}-\bar {\mb{D}}^{*}\right) \mb{z} \nonumber\\
= \hspace{0.05in}& \left (D_1-2^{-2 \left(R_1^{*} - \bar R_1(\bar {\mb{D}}) \right)} \right) \left(y + \frac{\mb{z}^T \mb{b}}{D_1-2^{-2 \left(R_1^{*} - \bar R_1(\bar {\mb{D}}) \right)}}\right)^2 \nonumber\\
&+ \mb{z}^T \left ( \bar {\mb{D}}-\bar {\mb{D}}^{*} - \frac{\mb{b} \mb{b}^T}{D_1-2^{-2 \left(R_1^{*} - \bar R_1(\bar {\mb{D}}) \right)}}\right) \mb{z}\\
\ge \hspace{0.05in}&0,
\end{align}
where (71) and (72) follow from the conditions in Lemma 6(c). This implies that (70) holds, and hence we have the equality in (67). This completes the proof of Lemma 6.
\end{proof}
\section{Conclusion}
We determined the rate region of the Gaussian scalar-help-vector source-coding problem, and proved that the Gaussian achievable scheme is optimal for the problem. We introduced a novel way of establishing the converse. Our approach involved lower bounding the problem with a potentially reduced dimensional problem by projecting the main source and imposing the distortion constraints in certain directions determined by the optimal Gaussian scheme. The core optimization problem for the converse turns out to be the point-to-point rate-distortion problem for a vector Gaussian source under a covariance matrix distortion constraint. The properties satisfied by the optimal solution to the point-to-point problem plays an important role in our converse arguments. We also generalized our work to the general source coding problem in which there are distortion constraints on both the sources, and obtained an outer bound to the rate region. The outer bound is partially tight in general. We also studied its tightness in some nontrivial cases. \\ \\
\hspace{0.1in}\\
\Large
{\textbf{Appendix A}\hspace{0.1in} }\newline
\normalsize \\
We will use a transformation similar to the one used by Liu and Viswanath \cite{Lui}. Let us suppose that the rank of $\mb{K}_{\mb{X}}$ is $p < m$.  The spectral decomposition of $\mb{K}_{\mb{X}}$ is
\[
\mb{K}_{\mb{X}} = \mb{Q \Sigma}\mb{Q}^T,
\]
where $\mb{\Sigma}$ is an orthogonal matrix and
\[
\mb{\Sigma} = \textrm{Diag}(\alpha_1,\dots,\alpha_p,0,\dots,0)
\]
is a diagonal matrix. Let
\begin{align*}
\mb{Q}^T \mb{D Q} &= \mat{\mb{A} & \mb{B}^T \\ \mb{B} & \mb{C}},
\end{align*}
where $\mb{A}$, $\mb{B}$ and $\mb{C}$ are submatrices of dimensions $p \times p$, $ (m-p) \times p$ and $(m-p) \times (m-p)$, respectively and let
\begin{align*}
\mb{T_X} &= \mat{\mb{I}_p & -\mb{B}^T \mb{C}^{-1} \\ \mb{0} & \mb{I}_{m-p}} \mb{Q}^T.
\end{align*}
Then the transformation is give by
\begin{align*}
\bar {\mb{X}} = \mat{\tilde {\mb{X}}\\ \check {\mb{X}}} = \mb{T_X} \mb{X},
\end{align*}
where $\tilde{\mb{X}}$ is a $p$-dimensional random vector. The covariance matrix of $\bar {\mb{X}}$ is
\begin{align*}
\mb{K}_{\bar {\mb{X}}} &= \mb{T_X} \mb{K}_{\mb{X}} \mb{T}^T_{\mb{X}} \\
&=\mat{\mb{I}_p & -\mb{B}^T \mb{C}^{-1} \\ \mb{0} & \mb{I}_{m-p}} \mb{Q}^T \mb{K}_{\mb{X}} \mb{Q} \mat{\mb{I}_p & \mb{0} \\-\mb{C}^{-1}\mb{B} & \mb{I}_{m-p}} \\
&= \mat{\mb{I}_p & -\mb{B}^T \mb{C}^{-1} \\ \mb{0} & \mb{I}_{m-p}} \mb{\Sigma} \mat{\mb{I}_p & \mb{0} \\-\mb{C}^{-1}\mb{B} & \mb{I}_{m-p}} \\
&=\mb{\Sigma},
\end{align*}
which means that $\mb{K}_{\check {\mb{X}}} = \mb{0}$, and hence $\check {\mb{X}}$ is deterministic. Now,
\begin{align}
\mb{T_X} \mb{D} \mb{T}^T_{\mb{X}} &= \mat{\mb{I}_p & -\mb{B}^T \mb{C}^{-1} \\ \mb{0} & \mb{I}_{m-p}} \mb{Q}^T \mb{D} \mb{Q} \mat{\mb{I}_p & \mb{0} \\-\mb{C}^{-1}\mb{B} & \mb{I}_{m-p}} \nonumber\\
&= \mat{\mb{I}_p & -\mb{B}^T \mb{C}^{-1} \\ \mb{0} & \mb{I}_{m-p}}  \mat{\mb{A} & \mb{B}^T \\ \mb{B} & \mb{C}} \mat{\mb{I}_p & \mb{0} \\-\mb{C}^{-1}\mb{B} & \mb{I}_{m-p}} \nonumber\\
&= \mat{\mb{A}-\mb{B}^T \mb{C}^{-1}\mb{B} & \mb{0} \\\mb{0} & \mb{C}}.
\end{align}
Since $\mb{T_X}$ is invertible, the distortion constraint is equivalent to
\begin{align}
\mb{T_X} \mb{D} \mb{T}^T_{\mb{X}} &\succcurlyeq \frac{1}{n} \sum_{i=1}^n E\left [ \left (\bar {\mb{X}}_i - \hat {\bar{\mb{X}}}_i \right  ) \left (\bar{\mb{X}}_i - \hat {\bar{ \mb{X}}}_i \right )^T \right ] \nonumber\\
&=\frac{1}{n} \sum_{i=1}^n E\left [ \mat{\tilde {\mb{X}}_{i} - \hat {\tilde{\mb{X}}}_{i} \\ \mb{0}} \mat{\tilde {\mb{X}}_{i} - \hat {\tilde{\mb{X}}}_{i} \\ \mb{0}}^T \right] \nonumber\\
&=\mat{\frac{1}{n} \sum_{i=1}^n E\left [ \left (\tilde {\mb{X}}_{i} - \hat {\tilde{\mb{X}}}_{i} \right) \left (\tilde {\mb{X}}_{i} - \hat {\tilde{\mb{X}}}_{i} \right)^T \right] & \mb{0} \\ \mb{0} & \mb{0}}.
\end{align}
Since $\mb{A}-\mb{B}^T \mb{C}^{-1}\mb{B}$ and $\mb{C}$ are strictly positive definite, (73) and (74) imply that the equivalent distortion constraint is
\[
\mb{A}-\mb{B}^T \mb{C}^{-1}\mb{B} \succcurlyeq \frac{1}{n} \sum_{i=1}^n E\left [ \left (\tilde {\mb{X}}_{i} - \hat {\tilde{\mb{X}}}_{i} \right) \left (\tilde {\mb{X}}_{i} - \hat {\tilde{\mb{X}}}_{i} \right)^T \right].
\]
Finally, since $\check {\mb{X}}^n$ is deterministic, the first encoder can send information about $\tilde{\mb{X}}^n$ only. We therefore have an equivalent Gaussian scalar-help-vector source-coding problem in which the covariance matrix of the first encoder's observations
\[
\mb{K}_{\tilde {\mb{X}}} = \textrm{Diag}(\alpha_1,\dots,\alpha_p)
\]
is positive definite. Hence, we can assume without loss of generality that $\mb{K}_{{\mb{X}}}$ is strictly positive definite.\\ \\
\Large
{\textbf{Appendix B:}\hspace{0.1in} \textbf{Proof of Lemma 2}}\newline
\normalsize \\
Assume $(R_1,R_2,\mb{D})$ is achievable. Let $C_1 = f_1^{(n)} \left (\mb{X}^n \right )$ and $C_2=f_2^{(n)} (Y^n).$ We now have
\begin{align}
n R_2 &\ge \log M_2^{(n)}\nonumber\\
&\ge H(C_2) \nonumber\\
&\ge I(Y^n ; C_2) \nonumber \\
&= \sum_{i=1}^n \left [ h(Y_i | Y^{i-1}) - h(Y_i | Y^{i-1}, C_2)\right] \nonumber\\
&= \sum_{i=1}^n \left [ h(Y_i) - h(Y_i | \mb{X}^{i-1}, Y^{i-1}, C_2)\right] \\
&\ge \sum_{i=1}^n \left [ h(Y_i) - h(Y_i | \mb{X}^{i-1}, C_2)\right] \\
&= \sum_{i=1}^n \left [ h(Y_i) - h(Y_i | V_i)\right] \\
&= \sum_{i=1}^n I(Y_i;V_i),
\end{align}
where
\begin{enumerate}
\item[(75)] follows from the Markov condition $Y_i \leftrightarrow (Y^{i-1}, C_2) \leftrightarrow \mb{X}^{i-1}$,
\item[(76)] follows because conditioning reduces differential entropy, and
\item[(77)] follows by letting $V_i = (\mb{X}^{i-1}, C_2)$.
\end{enumerate}
Furthermore
\begin{align}
n R_1&\ge \log M_1^{(n)}\nonumber\\
&\ge H(C_1) \nonumber\\
&\ge H(C_1|C_2) \nonumber\\
&\ge I(\mb{X}^n;C_1|C_2) \nonumber\\
&= \sum_{i=1}^n I(\mb{X}_i;C_1|\mb{X}^{i-1},C_2) \nonumber\\
&= \sum_{i=1}^n I(\mb{X}_i;C_1|V_i).
\end{align}
Let $Q$ be a time-sharing random variable uniformly distributed over $\{ 1,\dots,n\}$ and independent of all other random variables and random vectors. Let us define
\begin{align*}
U &= (Q,C_1) \\
V &= (Q,V_Q) \\
Y &= Y_Q \\
\mb{X} &= \mb{X}_Q.
\end{align*}
Then we have from (78) that
\begin{align}
R_2 &\ge \frac{1}{n}\sum_{i=1}^n I(Y_i;V_i) \nonumber\\
&= I(Y_Q;V_Q|Q) \nonumber\\
&= I(Y_Q;V_Q,Q) - I(Y_Q;Q) \nonumber\\
&= I(Y_Q;V_Q,Q) \nonumber\\
&= I(Y;V)\nonumber
\end{align}
and from (79) that
\begin{align}
R_1&\ge \frac{1}{n} \sum_{i=1}^n I(\mb{X}_i;C_1|V_i) \nonumber\\
&= I(\mb{X}_Q;C_1|V_Q,Q) \nonumber\\
&= I(\mb{X}_Q;C_1,Q|V_Q,Q) \nonumber\\
&= I(\mb{X};U|V)\nonumber.
\end{align}
Now for each $i \in \{1,\dots,n\}$, we have a Markov chain
\[
\mb{X}_i = \mb{a} Y_i + \mb{N}_i \leftrightarrow Y_i \leftrightarrow V_i = (\mb{X}^{i-1}, C_2)
\]
because $\mb{N}_i$ is independent of $(\mb{X}^{i-1}, C_2)$. Thus, $\mb{X}, Y$ and $V$ also form a Markov chain
\begin{align*}
\mb{X}\leftrightarrow Y \leftrightarrow V.
\end{align*}
Finally, we have
\begin{align}
\mb{D} &\succcurlyeq \frac{1}{n} \sum_{i=1}^n E \left [(\mb{X}_i - \hat {\mb{X}}_i)(\mb{X}_i - \hat {\mb{X}}_i)^T \right ] \nonumber\\
&= \frac{1}{n} \sum_{i=1}^n E \left [(\mb{X}_i - E(\mb{X}_i|C_1,C_2))(\mb{X}_i - E(\mb{X}_i|C_1,C_2))^T \right ] \nonumber\\
&\succcurlyeq \frac{1}{n} \sum_{i=1}^n E \left [(\mb{X}_i - E(\mb{X}_i|C_1,C_2,\mb{X}^{i-1}))(\mb{X}_i - E(\mb{X}_i|C_1,C_2,\mb{X}^{i-1}))^T \right ] \\
&= \frac{1}{n} \sum_{i=1}^n E \left [(\mb{X}_i - E(\mb{X}_i|C_1,V_i))(\mb{X}_i - E(\mb{X}_i|C_1,V_i))^T \nonumber\right ]\\
&=E \left [(\mb{X}_Q - E(\mb{X}_Q|C_1,V_Q))(\mb{X}_Q - E(\mb{X}_Q|C_1,V_Q))^T \right ] \nonumber\\
&=E \left [(\mb{X}_Q - E(\mb{X}_Q|C_1,V_Q,Q))(\mb{X}_Q - E(\mb{X}_Q|C_1,V_Q,Q))^T \right ] \nonumber\\
&=E \left [(\mb{X} - E(\mb{X}|U,V))(\mb{X} - E(\mb{X}|U,V))^T \right ] \nonumber\\
&= \mb{K}_{\mb{X}|U,V} \nonumber,
\end{align}
where (80) follows because conditioning reduces the covariance of the error in a positive semidefinite sense. \\ \\
\Large
{\textbf{Appendix C:}\hspace{0.1in} \textbf{Proof of Lemma 3}}\newline
\normalsize \\
First note that if the second encoder's message alone is sufficient to meet the distortion constraint, then $R_1 = 0$. Hence, the second encoder can use all available rate $R_2$ for the transmission of a message to the decoder and therefore $\sigma^2_{Y|V^{*}} = \sigma^2_{Y} 2^{-2 R_2}$ is optimal in this case. Consider now the other case when both encoders need to send messages to the decoder. We will first show that if we restrict the solution space of $P$ to feasible jointly Gaussian distributions, then we can assume without loss of generality that we have the long Markov chain
\begin{align*}
U \leftrightarrow \mb{X} \leftrightarrow Y  \leftrightarrow  V.
\end{align*}
It suffices to show the same for Gaussian $\bar U$ and $\bar V$ such that $\sigma^2_{Y|\bar V}$ is feasible and $\mb{K}_{\mb{X}|\bar U, \bar V}$ is the corresponding optimal solution to problem $F \left (\mb{D}, \mb{K}_{\mb{X}|\bar V} \right )$ defined in (1). Then from (34), $\left(\sigma^2_{Y|\bar V},\mb{K}_{\mb{X}|\bar U, \bar V} \right)$ is a candidate to be an optimal solution to $P_G$. From Section 4, $\mb{K}_{\mb{X}|\bar U, \bar V}$ gives two sets of directions $\mb{S}$ and $\mb{T}$ which satisfy the properties in Theorem 2. Let us define
\begin{align*}
\hat U = \mb{S}^{T} \mb{X} + \mb{W},
\end{align*}
where $\mb{W}$ is a zero-mean Gaussian random vector independent of $\mb{X}$ and has a covariance matrix $\mb{K_W}$ that satisfies
\begin{align}
\left (\mb{S}^T \mb{K}_{\mb{X}|\bar U, \bar V} \mb{S} \right )^{-1} = \left (\mb{S}^T \mb{K}_{\mb{X}|\bar V} \mb{S} \right )^{-1} + \mb{K}^{-1}_{\mb{W}}.
\end{align}
Note that $\mb{K_W}$ is strictly positive definite because
\[
\mb{S}^T \mb{K}_{\mb{X}|\bar U, \bar V} \mb{S}  \prec \mb{S}^T \mb{K}_{\mb{X}|\bar V} \mb{S}
\]
from Theorem 2(a). The conditional covariance of $\mb{X}$ given $(\hat U, \bar V)$ can be expressed as
\begin{align*}
\mb{K}_{\mb{X}|\hat U,\bar V} &=\mb{K}_{\mb{X}|\bar V} - E\left (\mb{X} \hat U^{T} | \bar V\right) \mb{K}^{-1}_{\hat U|\bar V} E\left ( \hat U \mb{X}^{T} | \bar V\right) \\
&= \mb{K}_{\mb{X}|\bar V} - \mb{K}_{\mb{X}|\bar V} \mb{S} \left(\mb{S}^T \mb{K}_{\mb{X}|\bar V} \mb{S}+\mb{K_W} \right)^{-1} \mb{S}^{T} \mb{K}_{\mb{X}|\bar V},
\end{align*}
which implies
\begin{align}
\mb{S}^T \mb{K}_{\mb{X}|\hat U,\bar V} \mb{S} &= \mb{S}^T \mb{K}_{\mb{X}|\bar V} \mb{S} - \mb{S}^T \mb{K}_{\mb{X}|\bar V} \mb{S} \left(\mb{S}^T \mb{K}_{\mb{X}|\bar V} \mb{S}+\mb{K_W} \right)^{-1} \mb{S}^T \mb{K}_{\mb{X}|\bar V} \mb{S} \nonumber\\
&= \left(\left(\mb{S}^T \mb{K}_{\mb{X}|\bar V} \mb{S}\right)^{-1}+\mb{K}^{-1}_{\mb{W}} \right)^{-1} \nonumber\\
&= \mb{S}^T \mb{K}_{\mb{X}|\bar U,\bar V} \mb{S},
\end{align}
where (82) follows from (81), and
\begin{align}
\mb{T}^T \mb{K}_{\mb{X}|\hat U,\bar V} \mb{T} &= \mb{T}^T \mb{K}_{\mb{X}|\bar V} \mb{T} - \mb{T}^T \mb{K}_{\mb{X}|\bar V} \mb{S} \left(\mb{S}^T \mb{K}_{\mb{X}|\bar V} \mb{S}+\mb{K_W} \right)^{-1} \mb{S}^T \mb{K}_{\mb{X}|\bar V} \mb{T} \nonumber\\
&= \mb{T}^T \mb{K}_{\mb{X}|\bar V} \mb{T} \\
&= \mb{I}_l \\
&= \mb{T}^T \mb{K}_{\mb{X}|\bar U,\bar V} \mb{T}, \\
\mb{T}^T \mb{K}_{\mb{X}|\hat U,\bar V} \mb{S} &= \mb{T}^T \mb{K}_{\mb{X}|\bar V} \mb{S} - \mb{T}^T \mb{K}_{\mb{X}|\bar V} \mb{S} \left(\mb{S}^T \mb{K}_{\mb{X}|\bar V} \mb{S}+\mb{K_W} \right)^{-1} \mb{S}^T \mb{K}_{\mb{X}|\bar V} \mb{S} \nonumber\\
&=\mb{0} \\
&= \mb{T}^T \mb{K}_{\mb{X}|\bar U,\bar V} \mb{S},
\end{align}
where
\begin{enumerate}
\item[(83)] and (86) follow because $\mb{S}$ and $\mb{T}$ are cross $\mb{K}_{\mb{X}|\bar V}$-orthogonal from Theorem 2(g),
\item[(84)] follows because $\mb{T}$ is $\mb{K}_{\mb{X}|\bar V}$-orthogonal from Theorem 2(e),
\item[(85)] follows because $\mb{T}$ is $\mb{K}_{\mb{X}|\bar U,\bar V}$-orthogonal from Theorem 2(c), and
\item[(87)] follows because $\mb{S}$ and $\mb{T}$ are cross $\mb{K}_{\mb{X}|\bar U,\bar V}$-orthogonal from Theorem 2(c).
\end{enumerate}
In summary, we have
\[
[\mb{S,T}]^T \mb{K}_{\mb{X}|\hat U,\bar V} [\mb{S,T}] = [\mb{S,T}]^T \mb{K}_{\mb{X}|\bar U,\bar V} [\mb{S,T}],
\]
and hence
\[
\mb{K}_{\mb{X}|\hat U,\bar V} = \mb{K}_{\mb{X}|\bar U,\bar V},
\]
because $[\mb{S,T}]$ is invertible from Theorem 2(b). This proves that we can assume without loss of generality that $\bar U$ is of the following form
\begin{align*}
\bar U &= \mb{S}^{T}\mb{X} + \mb{W}.
\end{align*}
and therefore we have the following long Markov chain
\begin{align*}
\bar U \leftrightarrow \mb{X} \leftrightarrow Y  \leftrightarrow  \bar V.
\end{align*}
Hence, without loss of generality, we can assume that any feasible solution to $P_G$, in particular $(U^{*},V^{*})$ satisfies the long Markov chain. Next, because of the second encoder's rate constraint, we have
\begin{align}
\sigma^2_{Y|V^{*}} \ge \sigma^2_{Y} 2^{-2 R_2}.
\end{align}
We want to prove equality in (88). Suppose otherwise that
\begin{align}
\sigma^2_{Y|V^{*}} > \sigma^2_{Y} 2^{-2 R_2}.
\end{align}
Then there exists a zero-mean Gaussian random variable $\tilde V$ such that for some $\epsilon > 0$, the conditional variance of $Y$ given $\tilde V$
\begin{align*}
\sigma^2_{Y|\tilde V} = \sigma^2_{Y|V^{*}} - \epsilon > \sigma^2_{Y} 2^{-2 R_2},
\end{align*}
and $U^{*}$, $\mb{X}$, $Y$, $\tilde V$ and $V^{*}$ form a Markov chain
\begin{align}
U^{*} \leftrightarrow \mb{X} \leftrightarrow Y \leftrightarrow \tilde V \leftrightarrow V^{*}.
\end{align}
We therefore have
\begin{align*}
\mb{K}_{\mb{X}|U^{*},\tilde V} &=\mb{K}_{\mb{X}|U^{*},\tilde V,V^{*}}\\
&\preccurlyeq \mb{K}_{\mb{X}|U^{*},V^{*}} \\
&\preccurlyeq \mb{D} \\
\textrm{and} \hspace{2in}\\
\mb{K}_{\mb{X}|U^{*},\tilde V} &\preccurlyeq  \mb{K}_{\mb{X}|\tilde V} \hspace{2in}\\
&= \mb{a}\mb{a}^{T} \sigma^2_{Y|\tilde V} + \mb{K_N},
\end{align*}
which means that $\left (\mb{K}_{\mb{X}|U^{*},\tilde V},\sigma^2_{Y|\tilde V}\right )$ is feasible for problem $P_G$. We now have the following chain of inequalities
\begin{align}
I(U^{*};\mb{X} |\tilde V ) &= I(U^{*};\mb{X} |\tilde V, V^{*}) \\
&= I(U^{*};\mb{X},\tilde V | V^{*}) - I(U^{*};\tilde V | V^{*}) \nonumber\\
&< I(U^{*};\mb{X},\tilde V | V^{*}) \\
&= I(U^{*};\mb{X} | V^{*})+I(U^{*};\tilde V | V^{*},\mb{X}) \nonumber\\
&= I(U^{*};\mb{X} | V^{*}),
\end{align}
where
\begin{enumerate}
\item[(91)] and (93) follows from the Markov condition in (90), and
\item[(92)] follows because
\begin{align*}
I(U^{*};\tilde V | V^{*}) &= \frac{1}{2} \log \frac{\left|\mb{K}_{U^{*}| V^{*}} \right|}{\left|\mb{K}_{U^{*}| V^{*}, \tilde V} \right|} \\
&= \frac{1}{2} \log \frac{\left|\mb{K}_{U^{*}| V^{*}} \right|}{\left|\mb{K}_{U^{*}| \tilde V} \right|} \\
&= \frac{1}{2} \log \frac{\left|\mb{S}^T \mb{K}_{\mb{X}| V^{*}} \mb{S} + \mb{K}_{\mb{W}}\right|}{\left|\mb{S}^T \mb{K}_{\mb{X}| \tilde V} \mb{S} + \mb{K}_{\mb{W}} \right|} \\
&= \frac{1}{2} \log \frac{\left|\mb{S}^T \left (\mb{a}\mb{a}^T \sigma^2_{Y|V^{*}}+\mb{K}_{\mb{N}} \right ) \mb{S} + \mb{K}_{\mb{W}}\right|}{\left|\mb{S}^T \left (\mb{a}\mb{a}^T \sigma^2_{Y|\tilde V}+\mb{K}_{\mb{N}} \right ) \mb{S} + \mb{K}_{\mb{W}} \right|}\\
&> 0.
\end{align*}
\end{enumerate}
We have arrived at a contradiction to the assumption that $U^{*}$ and $V^{*}$ are the optimal Gaussian random variables. Therefore, the supposition (89) is wrong, and hence (88) holds with equality. \\ \\
\Large
{\textbf{Appendix D:}\hspace{0.1in} \textbf{Proof of Lemma 4}}\newline
\normalsize \\
We have
\begin{align}
h\left (\mb{S}^T\mb{X} | U, V \right ) &\le \frac{1}{2} \log \left ( \left (2 \pi e \right)^r \left |\mb{S}^T\mb{K}_{\mb{X}|U,V} \mb{S} \right| \right ) \\
&\le \frac{1}{2} \log \left ( \left (2 \pi e \right)^r \left |\mb{S}^T \mb{D} \mb{S} \right| \right ),
\end{align}
where
\begin{enumerate}
\item[(94)] follows from the fact the Gaussian distribution maximizes the differential entropy for a given covariance matrix, and
\item[(95)] follows from the distortion constraint.
\end{enumerate}
The inequalities (94) and (95) are tight if $\mb{X}, U$, and $V$ are jointly with the conditional covariance matrix $\mb{K}_{\mb{X}|U,V}$ such that
\begin{align}
\mb{S}^T \mb{K}_{\mb{X}|U,V} \mb{S} = \mb{S}^T \mb{D} \mb{S}.
\end{align}
Since $\mb{K}_{\mb{X}|U^{*},V^{*}}$ satisfies (96), we conclude that a Gaussian $(U,V)$ with the conditional covariance matrix $\mb{K}_{\mb{X}|U^{*},V^{*}}$ is optimal for subproblem $\tilde P(\mb{D})$, and the optimal value is
\begin{align}
v\left (\tilde P(\mb{D}) \right ) &= h\left (\mb{S}^T\mb{X} \right ) - \frac{1}{2} \log \left ( \left (2 \pi e \right)^r \left |\mb{S}^T \mb{D} \mb{S} \right| \right ) \nonumber\\
&= \frac{1}{2} \log \left ( \left (2 \pi e \right)^r \left | \mb{S}^T\mb{K_X} \mb{S} \right | \right ) - \frac{1}{2} \log \left ( \left (2 \pi e \right)^r \left |\mb{I}_r \right| \right ) \\
&= \frac{1}{2} \log \left | \mb{S}^T\mb{K_X} \mb{S} \right | \nonumber,
\end{align}
where (97) follows because $\mb{S}$ is $\mb{D}$-orthogonal from Theorem 2(d).\\  \\
\Large
{\textbf{Appendix E:}\hspace{0.1in} \textbf{Proof of Lemma 5}}\newline
\normalsize \\
First note that if
\begin{align*}
\mb{S}^T \mb{a} = \mb{0},
\end{align*}
then
\begin{align*}
\mb{S}^T \mb{X} = \mb{S}^T \left (\mb{a}Y + \mb{N} \right )= \mb{S}^T \mb{N},
\end{align*}
which means that
\begin{align*}
v(P(R_2)) = 0,
\end{align*}
because $Y$ is independent of $\mb{N}$, and we have a Markov condition $\mb{S}^T \mb{X} \leftrightarrow Y \leftrightarrow V$. So, any $V$ including a Gaussian one with the conditional variance $\sigma^2_{Y|V^{*}}$ is optimal for subproblem $\tilde P(R_2)$. Therefore, Lemma 5 is trivially true in this case. Let us assume now that
\begin{align*}
\mb{S}^T \mb{a} \neq \mb{0},
\end{align*}
and let
\[
\mb{u}_1, \mb{u}_2, \dots, \mb{u}_r
\]
be an orthonormal basis in $\mathbb{R}^r$ starting at
\begin{align*}
\mb{u}_1 =  \frac{1}{c} \left (\mb{S}^T \mb{K_X} \mb{S} \right )^{-1/2} \mb{S}^T \mb{a},
\end{align*}
where
\[
c = \left \| \left (\mb{S}^T \mb{K_X} \mb{S} \right )^{-1/2} \mb{S}^T \mb{a} \right \|.
\]
Define the matrices
\begin{align*}
\mb{U} &= \left [\mb{u}_1, \mb{u}_2, \dots, \mb{u}_r \right ]\\
\mb{T}_{\mb{X}} &= \mb{U}^T \left (\mb{S}^T \mb{K_X} \mb{S} \right )^{-1/2},
\end{align*}
and the transformation
\[
\tilde {\mb{X}} = \mb{T}_{\mb{X}} \left (\mb{S}^T \mb{X} \right).
\]
Then the covariance matrix of $\tilde {\mb{X}}$ is
\begin{align*}
\mb{K}_{\tilde {\mb{X}}} &=  \mb{T}_{\mb{X}} \left (\mb{S}^T \mb{K_X} \mb{S} \right )\mb{T}_{\mb{X}}^T\\
&= \mb{U}^T \left (\mb{S}^T \mb{K_X} \mb{S} \right )^{-1/2} \left (\mb{S}^T \mb{K_X} \mb{S} \right ) \left (\mb{S}^T \mb{K_X} \mb{S} \right )^{-1/2} \mb{U}\\
&= \mb{U}^T \mb{U} \\
&= \mb{I}_r,
\end{align*}
and the cross-covariance matrix between $\tilde {\mb{X}}$ and $Y$ is
\begin{align*}
\mb{K}_{\tilde {\mb{X}}Y} &=\mb{T}_{\mb{X}} \mb{S}^T\mb{K}_{{\mb{X}}Y} \\
&=\mb{U}^T \left (\mb{S}^T \mb{K_X} \mb{S} \right )^{-1/2} \mb{S}^T \mb{a} \sigma^2_Y \\
&=\mb{U}^T \mb{u}_1 \left \| \left (\mb{S}^T \mb{K_X} \mb{S} \right )^{-1/2} \mb{S}^T \mb{a} \right \| \sigma^2_Y \\
&=\left (c \sigma^2_Y ,0, \dots, 0\right)^T.
\end{align*}
This means that under this transformation, $\tilde {\mb{X}}$ has i.i.d standard normal components, and $Y$ is correlated with $\tilde {{X}}_1$ only and is uncorrelated with the rest of the components of $\tilde {\mb{X}}$. Since the transformation matrix $\mb{T}_{\mb{X}}$ is full rank, we have
\begin{align}
I(\mb{S}^T \mb{X} ; V) &= I(\tilde {\mb{X}} ; V) \nonumber\\
&= I(\tilde {{X}}_1 ; V) + I(\tilde {{X}}_2, \dots, \tilde {{X}}_r ; V | \tilde {{X}}_1) \nonumber\\
&= I(\tilde {{X}}_1 ; V),
\end{align}
where (98) follows because $(\tilde {{X}}_2, \dots, \tilde {{X}}_r)$ is independent of $(V,\tilde {{X}}_1)$. It is also clear that $\tilde {{X}}_1$, $Y$ and $V$ form a Markov chain
\[
\tilde {{X}}_1 \leftrightarrow Y \leftrightarrow V.
\]
Therefore, the subproblem $\tilde P(R_2)$ is equivalent to the following problem
\begin{align*}
\max_{V} \hspace {0.1in} &I(\tilde {{X}}_1;V) \\
\textrm{subject to} \hspace{0.1in} &R_2 \ge I(Y;V) \\
&\tilde {{X}}_1 \leftrightarrow Y \leftrightarrow V,
\end{align*}
which is a well-known problem in scalar Gaussian lossy one-helper problem. Oohama in \cite{Oohama} showed that a Gaussian $V$ with the conditional variance $\sigma^2_{Y|V^{*}}$ is optimal for this problem. Hence, the same Gaussian solution is optimal for $\tilde P(R_2)$ too. Thus, the optimal solution to $\tilde P(R_2)$ is Gaussian with the conditional variance $\sigma^2_{Y|V^{*}}$ and the optimal value is
\begin{align}
v\left (\tilde P(R_2) \right) &= \frac{1}{2} \log \frac {\left | \mb{S}^T\mb{K_X} \mb{S} \right |}{\left | \mb{S}^T \left (\mb{a} \mb{a}^T  \sigma^2_{Y|V^{*}} + \mb{K_N} \right) \mb{S} \right |} \nonumber\\
&= \frac{1}{2} \log \frac {\left | \mb{S}^T\mb{K_X} \mb{S} \right |}{\left | \mb{S}^T\mb{K}_{\mb{X}|V^{*}} \mb{S} \right |},
\end{align}
where (99) follows from (37).

\end{document}